\documentclass[journal]{IEEEtran}

\usepackage{epsfig}
\usepackage{amsmath}
\usepackage{amssymb}
\usepackage{amsthm}
\usepackage{cite}
\usepackage{cleveref}
\usepackage[caption=false]{subfig}
\usepackage{xcolor}
\DeclareMathOperator*{\argmax}{arg\,max}
\DeclareMathOperator*{\argmin}{arg\,min}
\newtheorem{theorem}{Theorem}

\newtheorem{corollary}{Corollary}
\newtheorem{remark}{Remark}
\newtheorem{definition}{Definition}

\def\baselinestretch{0.9346}
\makeatletter
  \def\my@tag@font{\small}
  \def\maketag@@@#1{\hbox{\m@th\normalfont\my@tag@font#1}}
  \let\amsmath@eqref\eqref
  \renewcommand\eqref[1]{{\let\my@tag@font\relax\amsmath@eqref{#1}}}
\makeatother

\begin{document}
\title{A Moving Target Defense for Securing Cyber-Physical Systems}
\author{Paul~Griffioen,~\IEEEmembership{Student~Member,~IEEE,}
	Sean~Weerakkody,~\IEEEmembership{Member,~IEEE,}
	and~Bruno~Sinopoli,~\IEEEmembership{Senior~Member,~IEEE}%
\thanks{P. Griffioen and S. Weerakkody are with the Department of Electrical and Computer Engineering, Carnegie Mellon University, Pittsburgh, PA, USA 15213. B. Sinopoli is with the Department of Electrical and Systems Engineering, Washington University in St. Louis, St. Louis, MO, USA 63130. Email: {\tt\small pgriffi1@andrew.cmu.edu, sweerakk@alumni.cmu.edu, bsinopoli@wustl.edu}}
\thanks{This work is partially supported by Department of Energy grant DE-OE0000779 and by National Science Foundation grant 164652.}%
}


\maketitle

\begin{abstract}
This article considers the design and analysis of multiple moving target defenses for recognizing and isolating attacks on cyber-physical systems. We consider attackers who perform integrity attacks on some set of sensors and actuators in a control system. In such cases, it has been shown that a model aware adversary can carefully design attack vectors to bypass bad data detection and identification filters while causing damage to the control system. To counter such an attacker, we propose the moving target defense which introduces stochastic, time-varying parameters in the control system. The underlying random dynamics of the system limit an attacker's knowledge of the model and inhibits his or her ability to construct stealthy attack sequences. Moreover, the time-varying nature of the dynamics thwarts adaptive adversaries. We explore three main designs.  First, we consider a hybrid system where parameters within the existing plant are switched among multiple modes. We demonstrate how such an approach can enable both the detection and identification of malicious nodes. Next, we investigate the addition of an extended system with dynamics that are coupled to the original plant but do not affect system performance. Here, an attack on the original system will affect the authenticating subsystem and in turn be revealed by a set of sensors measuring the extended plant. Lastly, we propose the use of sensor nonlinearities to enhance the effectiveness of the moving target defense. The nonlinear dynamics act to conceal normal operational behavior from an attacker who has tampered with the system state, further hindering an attacker's ability to glean information about the time-varying dynamics. In all cases mechanisms for analysis and design are proposed. Finally, we analyze attack detectability for each moving target defense by investigating expected lower bounds on the detection statistic. Our contributions are also tested via simulation.
\end{abstract}

\IEEEpeerreviewmaketitle

\section{Introduction}
Securing cyber-physical systems (CPSs), the amalgamation of sensing, processing, control, and communication in physical spaces, is an essential goal in today's society. CPSs are ubiquitous in modern critical infrastructure such as transportation systems, energy delivery, health care, and sewage/water management. Consequently, these systems are attractive targets for adversaries and are essential to protect. Unfortunately, CPSs are vulnerable to adversarial attacks \cite{cardenas2008secure} due to the large number of attack surfaces found in these large scale, heterogeneous, and highly connected systems. Additionally, existing defenses from cyber security alone are insufficient for protecting CPSs. Traditional techniques such as authenticated encryption, message authentication codes, and signatures that typically enable the detection of integrity attacks can be computationally complex and are ineffective against a class of attacks known as physical attacks. Moreover, updating legacy systems can prove to be impractical.

The vulnerabilities in CPSs have culminated in several effective attacks from highly resourceful and knowledgeable attackers. In the year 2000, a malicious insider was able to utilize detailed system knowledge to attack a waste management system in Queensland, Australia \cite{slay2007lessons}, resulting in the leakage of millions of liters of sewage. With Stuxnet \cite{chen2010stuxnet}, a nation state adversary was able to compromise a uranium enrichment facility in Iran, leading to the destruction of a thousand centrifuges. More recently in 2015, hackers were able to remotely compromise a supervisory control and data acquisition (SCADA) system in Ukraine \cite{case2016analysis}, allowing them to cause widespread blackouts.

Motivated by the threat of such sophisticated attackers, we aim to design resilient CPSs. As a first step we focus on the problem of detecting and in some cases isolating attacks from malicious attackers. The problem of recognizing attacks is not trivial, especially when considering highly knowledgeable and resourceful attackers. For instance adversaries can utilize model knowledge to engage in deceptive and powerful stealthy attacks, including false data injection attacks \cite{liu2011false,mo2010false}, covert attacks \cite{smith2015covert}, zero dynamics attacks \cite{pasqualetti2013attack,teixeira2015secure}, and replay attacks \cite{mo2009secure}. Here the adversary is able to leverage access to system channels and/or model knowledge to construct attacks which bypass traditional bad data detectors such that the outputs received by a SCADA operator are statistically consistent with expected output behavior.

To counter such an attacker, a defender must engage in active detection \cite{weerakkody2017active} by designing a system that adds additional redundancy or introduces a physical secret. For instance, physical watermarking was introduced in \cite{mo2009secure} to counter replay adversaries. Here, the defender changes his or her control input to introduce random authenticating perturbations to the system. Several extensions have been pursued, for instance \cite{mo2014detecting,mo2015physical,satchidanandan2017dynamic,weerakkody2017bernoulli}. Alternatively, the defender can pursue one time changes to the system, including changes to the parameters \cite{teixeira2012revealing} or structural changes, for instance involving sensing and communication \cite{weerakkody2017robust}. In addition, encryption or lower cost mechanisms such as coding \cite{miao2014coding} can be effective tools for authentication. Nonetheless, the above schemes can be rendered ineffective by strong attackers. Watermarking can fail against additive attacks pursued by model aware attackers. Increasing robustness through one time changes can fail against highly resourceful attackers. Finally sensor coding can be ineffective against attackers with physical access to sensors and a certain class of zero dynamics attacks.

To address these challenges, we consider the moving target defense, which was first introduced in \cite{movingtarget} with extensions in \cite{switchedsystem,weerakkody2016moving,acc2017griffioen}. Here, the defender introduces time-varying parameters into the control system, resulting in periodic changes to the system matrices. The unknown parameters limit the attacker's understanding of the system model. Moreover, the time-varying dynamics ideally act as a moving target, changing fast enough to hinder a potential adaptive adversary from performing system identification.

The article considers three main moving target designs. In the first design, we evaluate a hybrid moving target where the system is switched among a number of discrete modes. We provide a set of design recommendations for the hybrid moving target which enable a defender to both detect and identify sensor attacks in control systems. Secondly, we design an extended moving target where we introduce an auxiliary system with time-varying dynamics coupled to the original plant. An attacker who perturbs the original system will also affect the additional dynamics due to this coupling. Moreover, the time-varying behavior of the system prevents the defender from concealing his or her attack through fake sensor measurements. We provide efficiently solvable optimization problems to design the parameters that generate the time-varying matrices in this extended moving target.

Thirdly, we note that even in the presence of time-varying dynamics, the attacker still has some opportunity to learn useful information about the model which can be applied to an attack. To limit this information, we introduce random nonlinearities in the sensor measurements which are amplified when the system state is perturbed and consequently conceal information about the system from the adversary when the plant is under attack. We provide a limit analysis to demonstrate the effectiveness of this approach as well as optimization problems to design the coefficient matrix associated with the nonlinearity. Lastly, we provide mechanisms to analyze attack detectability by investigating expected lower bounds on the detection statistic for each moving target defense.

An overarching goal of this article is to present the moving target as a general technique, which can be realized to counter a variety of attackers or molded to fit a variety of architectures. In order to illustrate this, previous results on the hybrid moving target defense \cite{weerakkody2016moving} and the extended moving target defense \cite{movingtarget,acc2017griffioen} are repeated in this article. The main contributions relative to our previous work in \cite{movingtarget}, \cite{weerakkody2016moving}, and \cite{acc2017griffioen} are as follows:
\begin{enumerate}
    \item An extension of the work in \cite{acc2017griffioen} to account for a time-varying covariance that generates the distribution of the auxiliary actuators. (section IV-A)
    \item The introduction, analysis, and design of random nonlinearities in the sensor measurements that conceal information about the time-varying dynamics from the adversary when the plant is under attack. (section V, VII-A)
    \item The presentation and validation of a simpler and more accurate method for computing expected lower bounds on the detection statistic than that given in \cite{movingtarget}. (section VI, VII-B)
    \item The organized presentation of three different moving target techniques differentiated by their differing architectures and the attack models they wish to thwart.
\end{enumerate}

The rest of the article is summarized below. In section II, we introduce the system and attack models along with the moving target defense. In section III, we consider the design of a hybrid moving target defense, placing a special focus on attack identification. Next, in section IV, we investigate the design of an extended moving target defense for attack detection. Later, in section V, we pursue the design of a nonlinear moving target defense to limit an attacker's ability to identify the system model.  In section VI, we propose statistical bounds to analyze the performance of the moving target defense. Lastly, section VII includes simulation results and section VIII concludes the article.

\section{Modeling the Moving Target}

\subsection{System Model}
To begin, we introduce the model for the system under consideration. We model our CPS as a linear time-invariant system as follows
\begin{equation} \small
x_{k+1} = Ax_k + Bu_k + w_k, \quad y_k = Cx_k + v_k.
\end{equation}
Here $x_k \in \mathbb{R}^n$ represents the system state at time $k$, $u_k \in \mathbb{R}^p$ is a vector of control inputs, and $y_k \in \mathbb{R}^m$ represents a collection of $m$ scalar sensor outputs. In addition, to capture uncertainty we consider independent and identically distributed (IID) Gaussian process noise $w_k \sim \mathcal{N}(0,Q)$ and IID Gaussian sensor noise $v_k \sim \mathcal{N}(0,R)$. We assume that $(A,B)$ and $(A,Q^{\frac{1}{2}})$ are stabilizable, $(A,C)$ is detectable, and $R \succ 0$.

In this article, we consider an adversary who can perform integrity attacks. For the hybrid moving target, we assume that the attacker is able to corrupt all of the outputs. For the extended moving target and nonlinear moving target, we assume that the attacker can corrupt all of the inputs and outputs. This for instance can be done over a network through a man in the middle attack where an attacker intercepts true packets and replaces them with false packets. Alternatively, physical attacks can disrupt the integrity of a system. For instance, the attacker can change the settings of programmable logic controllers (PLCs) or the environment surrounding system sensors. Mathematically, we model an integrity attack as follows
\begin{equation} \small
\label{IntegrityAttackStatesSensors}
x_{k+1} = Ax_k + Bu_k + B^au_k^a + w_k, \quad y_k = Cx_k + D^ad_k^a + v_k.
\end{equation}
Without loss of generality, an attack is assumed to begin at time $k = 0$. Here, $B^au_k^a$ represents attacks on the control inputs and $D^ad_k^a$ represents attacks on sensor outputs. If all actuators can be corrupted, $B^a = B$ and if all sensors can be modified, $D^ a = I$. Motivated by the resources of malicious insiders and nation state adversaries, we will in the case of the extended moving target and nonlinear moving target consider this worst case scenario. Additionally, we will also consider the possibility that an attacker has detailed system knowledge. A fundamental understanding of the plant when combined with significant disclosure and disruption resources can lead to powerful attacks \cite{teixeira2015secure}. For instance, an attacker can attempt to subtract his or her influence. Here an adversary chooses an arbitrary sequence of control inputs $\{u_k^a\}$ in order to drive the system along the controllable subspace $(A,B^a)$. To avoid detection, the attacker leverages model knowledge to construct stealthy outputs. Specifically, 
\begin{equation} \small
\label{ZeroDynamics}
D^a d_k^a = - C x_k^a, \quad x_{k+1}^a = Ax_k^a + B^a u_k^a, \quad x_0^a = 0.
\end{equation}
It can be shown that the probability distribution of the outputs under such an attack is identical to the distribution under normal operation. Consequently no standard bad data detector can recognize this adversarial behavior, and as a result this behavior is perfectly stealthy. We remark that a significant resource for an attacker here is model knowledge, which allows the adversary to carefully construct fake sensor outputs. In the ensuing subsections we propose three main designs which allow us to limit an attacker's knowledge of the system model. We call this collection of tools the moving target defense.

\subsection{Hybrid Moving Target Defense}
In the hybrid moving target we change parameters of the system, particularly the system matrices, in a time-varying fashion to limit the adversary's knowledge of the system model. The time-varying sequence of system matrices is known to the defender but kept hidden from the adversary, which limits the effectiveness of an adaptive attacker. The dynamics of the hybrid moving target are given below
\begin{equation} \small
\label{HybridSystemDynamics}
    x_{k+1} = A_k x_k + B_k u_k + w_k, \quad y_k = C_k x_k + v_k.
\end{equation}
We assume that our plant is a switching hybrid system. Here, $(A_k,B_k,C_k)$ belong to a finite set of modes $\Gamma = \{(A(1),B(1),C(1)), \cdots, (A(l),B(l),C(l))\}$. While $\Gamma$ may be known to an attacker, the exact realization of system matrices will be unknown. This forces an attacker to leverage imperfect system information when constructing an attack, which in turn can reveal his or her malicious behavior. We assume the adversary is able to modify all of the sensor measurements. The information available to the defender and the attacker at time step $k$, denoted by $\mathcal{I}_k^D$ and $\mathcal{I}_k^A$, respectively, is given by
\begin{equation*} \small
\begin{split}
\mathcal{I}_k^D &\triangleq \{A_{0:k},B_{0:k},C_{0:k},u_{0:k},y_{0:k}^a,f(w_k,v_k)\}, \\
\mathcal{I}_k^A &\triangleq \{\Gamma,u_{0:k},y_{0:k}^A,D^ad_{0:k}^a,f(w_k,v_k)\},
\end{split}
\end{equation*}
where $y_k^a$ represents the attacked sensor measurements received by the system operator, $y_k^A$ denotes the sensor measurements that the attacker intercepts, and $D^ad_{0:k}^a$ represents the bias the attacker adds to the sensor measurements.

\begin{remark}
The sequence of time-varying matrices can be determined for instance by a cryptographically secure pseudo random number generator (PRNG). Here, the sequence of system matrices will be entirely determined by the seed of the random number generator. As we aim to prevent an attacker from learning about the sequence of system matrices, the random seed must be hidden from the attacker. Alternatively, the defender must know the random seed to perform tasks of detection and estimation. As such the random seed serves as a root of trust and is analogous to a cryptographic key.
\end{remark}

In section III, we will investigate the applications of the hybrid moving target defense for the purposes of identifying sensor attacks in control systems, specifically considering how to design $\Gamma$ and the sequence of time-varying matrices. We note that introducing parameter changes to the system can result in tradeoffs between security and control performance. This issue is addressed in the next section when we consider the extended moving target defense.

\subsection{Extended Moving Target Defense}
In the extended moving target, an authenticating subsystem is added on top of the nominal control system. Specifically, we introduce additional states $\tilde{x}_k\in\mathbb{R}^{\tilde{n}}$ measured by additional sensors $\tilde{y}_k\in\mathbb{R}^{\tilde{m}}$ which have dynamics that are coupled to the dynamics of the original state $x_k$. The dynamics of the extended moving target are given below
\begin{align}
\label{ExtendedStateDynamics}
\text{\small$\underbrace{\begin{bmatrix}
\tilde{x}_{k+1} \\
x_{k+1}
\end{bmatrix}}_{\bar{x}_{k+1}}$} &\text{\small$=
\underbrace{\begin{bmatrix}
\tilde{A} & \bar{A}_k \\
0 & A
\end{bmatrix}}_{\mathcal{A}_k}
\underbrace{\begin{bmatrix}
\tilde{x}_k \\
x_k
\end{bmatrix}}_{\bar{x}_k} +
\underbrace{\begin{bmatrix}
\tilde{B}_k \\
B
\end{bmatrix}}_{\mathcal{B}_k}
u_k +
\underbrace{\begin{bmatrix}
\tilde{w}_k \\
w_k
\end{bmatrix}}_{\bar{w}_k},$} \\
\label{ExtendedSensorDynamics}
\text{\small$\underbrace{\begin{bmatrix}
\tilde{y}_k \\
y_k
\end{bmatrix}}_{\bar{y}_k}$} &\text{\small$=
\underbrace{\begin{bmatrix}
\tilde{C} & \bar{C}_k \\
0 & C
\end{bmatrix}}_{\mathcal{C}_k}
\underbrace{\begin{bmatrix}
\tilde{x}_k \\
x_k
\end{bmatrix}}_{\bar{x}_k} +
\underbrace{\begin{bmatrix}
\tilde{v}_k \\
v_k
\end{bmatrix}}_{\bar{v}_k},$}
\end{align}
with process noise $\bar{w}_k\sim\mathcal{N}(0,\mathcal{Q})$ and sensor noise $\bar{v}_k\sim\mathcal{N}(0,\mathcal{R})$ such that $\mathcal{Q}\triangleq\text{BlkDiag}(\tilde{Q},Q)\succeq0$ and $\mathcal{R}\triangleq\begin{bmatrix}\tilde{R}&\tilde{R}_{12}\text{ };&\tilde{R}_{12}^T&R\end{bmatrix}\succ0$. We assume that the time-varying matrices $\bar{A}_k$, $\tilde{B}_k$, and $\bar{C}_{k}$ are selected from an IID distribution (to be designed later in section IV). Without loss of generality, the control inputs are multiplexed to the actuators of both the nominal and extended systems. The extended moving target is designed so that if an adversary attempts to bias the original state $x_k$, he or she also modifies the auxiliary state $\tilde{x}_k$. This in turn will cause changes to the measurements $\tilde{y}_k$. Ideally an attacker who can modify $\tilde{y}_k$ will be unable to do so in a convincing manner due to his or her lack of knowledge about the time-varying dynamics. The time-varying behavior will also impede the task of system identification. We assume the adversary is able to modify all of the control inputs and sensor measurements. The information available to the defender and the attacker at time step $k$ is given by
\begin{equation*} \small
\medmuskip=3.09mu
\thinmuskip=3.09mu
\thickmuskip=3.09mu
\begin{split}
\mathcal{I}_k^D &\triangleq \{A,B,C,\tilde{A},\bar{A}_{0:k},\tilde{B}_{0:k},\tilde{C},\bar{C}_{0:k},u_{0:k},\bar{y}_{0:k}^a,f(\bar{w}_k,\bar{v}_k)\}, \\
\mathcal{I}_k^A &\triangleq \{A,B,C,\tilde{A},\tilde{C},f(\bar{A},\tilde{B},\bar{C}),u_{0:k},u_{0:k}^a,\bar{y}_{0:k}^A,\bar{d}_{0:k}^a,f(\bar{w}_k,\bar{v}_k)\},
\end{split}
\end{equation*}
where $\bar{y}_k^a$ represents the attacked sensor measurements received by the system operator, $u_k^a$ denotes the bias the attacker adds to the control inputs, $\bar{y}_k^A$ represents the sensor measurements that the attacker intercepts, and $\bar{d}_k^a$ denotes the bias the attacker adds to the sensor measurements.

\begin{remark}
While this system involves matrices $\bar{A}_k$, $\tilde{B}_k$, and $\bar{C}_k$ selected from an IID distribution, the extended moving target defense can still be effective in other scenarios. For instance, the system parameters can evolve at multiple time scales. In this case, the longer the target remains in place, the easier it is for the adversary to identify the system.
\end{remark}

A significant advantage of the extended moving target defense relative to the hybrid moving target defense is potential system performance. In particular, if we do not care about controlling the additional states $\tilde{x}_k$, the controller of the original system can remain unchanged and no online performance is sacrificed. Because the dynamics of the original plant remain in place, there is no tradeoff between security and control. We will consider the design of the parameters that generate the system matrices in the extended moving target for attack detection in section IV.

\begin{remark}
As before, the sequence of time-varying matrices can be determined by a PRNG. The extended system itself can be introduced by leveraging existing dynamics in the system. For instance one can consider the dynamics of waste products such as heat in a chemical reaction or the friction in a mechanical generator. The dynamics can be made time-varying by changing conditions at the plant. Alternatively, one can introduce external hardware such as RLC circuits with variable resistors and capacitors to generate the time-varying auxiliary system.
\end{remark}

\subsection{Nonlinear Moving Target Defense}
The utility of the moving target lies in the challenges it poses for an adversary aiming to perform system identification. However, we acknowledge that the sensor measurements as constructed do reveal some information about the system dynamics. In order to further limit the information available to an attacker, we can intentionally introduce nonlinearities. Many systems are inherently nonlinear, allowing us to leverage the dynamics of the system to introduce these nonlinearities. Notably, consider the sensor measurements for the nonlinear moving target below
\begin{equation} \small
\label{ExtendedNonlinearSensorDynamics}
\underbrace{\begin{bmatrix}
\tilde{y}_k \\
y_k
\end{bmatrix}}_{\bar{y}_k} =
\underbrace{\begin{bmatrix}
\tilde{C} & \bar{C}_k \\
0 & C
\end{bmatrix}}_{\mathcal{C}_k}
\underbrace{\begin{bmatrix}
\tilde{x}_k \\
x_k
\end{bmatrix}}_{\bar{x}_k} +
\begin{bmatrix} G_kh(x_k) \\ 0 \end{bmatrix} + 
\underbrace{\begin{bmatrix}
\tilde{v}_k \\
v_k
\end{bmatrix}}_{\bar{v}_k}.
\end{equation}
It is assumed that the extended state dynamics are unchanged \eqref{ExtendedStateDynamics}. However, a nonlinearity $G_kh(x_k)$ is introduced into the auxiliary sensor measurements where $G_k$ is a random matrix chosen from an IID distribution. Here the nonlinearity takes a form such that $G_k$ determines the direction of the nonlinearity and $h(x_k)$ determines the magnitude of the nonlinearity. The nonlinearity is designed so that it is approximately $0$ when the system state lies within a normal region of operation. When the state has been perturbed away from its normal region of operation, the nonlinearity becomes large and unpredictable. We assume the adversary is able to modify all of the control inputs and sensor measurements. The information available to the defender and the attacker at time step $k$ is given by
\begin{equation*} \small
\begin{split}
\mathcal{I}_k^D \triangleq \{&A,B,C,\tilde{A},\bar{A}_{0:k},\tilde{B}_{0:k},\tilde{C},\bar{C}_{0:k},G_{0:k},\text{nonlinear function }h, \\
&u_{0:k},\bar{y}_{0:k}^a,f(\bar{w}_k,\bar{v}_k)\}, \\
\mathcal{I}_k^A \triangleq \{&A,B,C,\tilde{A},\tilde{C},f(\bar{A},\tilde{B},\bar{C}),f(G),\text{nonlinear function }h, \\
&u_{0:k},u_{0:k}^a,\bar{y}_{0:k}^A,\bar{d}_{0:k}^a,f(\bar{w}_k,\bar{v}_k)\}.
\end{split}
\end{equation*}

An attacker who aims to remain stealthy must be able to produce counterfeit measurements which do not contain this large nonlinearity. Nonetheless, this is impractical because the attacker does not know the time-varying matrix $G_k$ which determines the nonlinearity. Moreover, the large highly nonlinear attack measurements will significantly impede an attacker's ability to learn the time-varying matrices $(\bar{A}_k, \tilde{B}_k, \bar{C}_k)$ from the measurements $\tilde{y}_k$. The design of matrix $G_k$ and an analysis of the nonlinear moving target is presented in section V.

\subsection{Estimation and Detection}
A Kalman filter can be used to compute the minimum mean squared error state estimate $\hat{\bar{x}}_{k|k}$ given the set of previous measurements up to $\bar{y}_k$. The Kalman filter is a linear estimator given by
\begin{align}
\label{KalmanAPriori}
\text{\small$\hat{\bar{x}}_{k+1|k}$} &\text{\small$= \mathcal{A}_k\hat{\bar{x}}_{k|k} + \mathcal{B}_ku_k,$} \\
\label{KalmanAPosteriori}
\text{\small$\hat{\bar{x}}_{k|k}$} &\text{\small$= (I-\mathcal{K}_k\mathcal{C}_k)\hat{\bar{x}}_{k|k-1} + \mathcal{K}_k\bar{y}_k,$} \\
\label{KalmanGain}
\text{\small$\mathcal{K}_k$} &\text{\small$= \mathcal{P}_{k|k-1}\mathcal{C}_k^T(\mathcal{C}_k\mathcal{P}_{k|k-1}\mathcal{C}_k^T+\mathcal{R})^{-1},$} \\
\label{ErrorCovariance}
\text{\small$\mathcal{P}_{k+1|k}$} &\text{\small$= \mathcal{A}_k(I-\mathcal{K}_k\mathcal{C}_k)\mathcal{P}_{k|k-1}\mathcal{A}_k^T + \mathcal{Q},$}
\end{align}
where $\hat{\bar{x}}_{k+1|k}$ is the a priori state estimate, $\hat{\bar{x}}_{k|k}$ is the a posteriori state estimate, $\mathcal{P}_{k+1|k}$ is the a priori error covariance matrix, and $\mathcal{K}_k$ is the Kalman gain. To detect attacks on the CPS, a residue-based detector that leverages the a priori state estimate $\hat{\bar{x}}_{k|k-1}$ is utilized. The residue $\bar{z}_k$ represents the difference between the observed and expected value of the measurements and is given by
\begin{equation} \small
\label{Residue}
\bar{z}_k = \bar{y}_k - \mathcal{C}_k\hat{\bar{x}}_{k|k-1}.
\end{equation}
By incorporating this residue, a $\chi^2$ detector given by
\begin{equation} \small
\label{Detector}
g_k(\bar{z}_{k-T+1:k}) = \sum_{i=k-T+1}^k\bar{z}_i^T(\mathcal{C}_i\mathcal{P}_{i|i-1}\mathcal{C}_i^T+\mathcal{R})^{-1}\bar{z}_i \mathop{\gtrless}_{\mathcal{H}_0}^{\mathcal{H}_1} \eta_k,
\end{equation}
with detection statistic $g_k$ follows a $\chi^2$ distribution under normal operation. The $\chi^2$ detector, which has $T(m+\tilde{m})$ degrees of freedom, attempts to exploit this fact by testing to see if the residues follow the correct distribution. Here $\eta_k$ represents the threshold of the bad data detector, $\mathcal{H}_0$ is the null hypothesis which represents normal system operation, $\mathcal{H}_1$ is the alternative hypothesis which denotes that the system is under attack, and $T$ represents the detector window that considers past measurements. Measurements that are in close agreement with expected values generate small detection statistics and thus raise no alarm. Large deviations between measured and expected behavior will lead to a large detection statistic, thus causing an alarm.

\begin{remark}
Estimation and detection for the hybrid moving target is described by replacing $\hat{\bar{x}}_{k+1|k}$, $\hat{\bar{x}}_{k|k}$, $\bar{y}_k$, $\bar{z}_k$, $\mathcal{A}_k$, $\mathcal{B}_k$, $\mathcal{C}_k$, $\mathcal{K}_k$, $\mathcal{P}_{k+1|k}$, $\mathcal{Q}$, and $\mathcal{R}$ in \cref{KalmanAPriori,KalmanAPosteriori,KalmanGain,ErrorCovariance,Residue,Detector} with $\hat{x}_{k+1|k}$, $\hat{x}_{k|k}$, $y_k$, $z_k$, $A_k$, $B_k$, $C_k$, $K_k$, $P_{k+1|k}$, $Q$, and $R$, respectively, where $\hat{x}_{k+1|k}$ and $\hat{x}_{k|k}$ are the a priori and a posteriori state estimates for the nominal system, $z_k$ is the residue for the nominal system, $K_k$ is the Kalman gain for the nominal system, and $P_{k+1|k}$ is the a priori error covariance matrix for the nominal system.
\end{remark}

While the estimation and detection techniques described above can be applied to the hybrid moving target and the extended moving target, a slight modification must occur when performing estimation and detection for the nonlinear moving target. Because the sensor measurements are nonlinear, an extended Kalman filter is used and is given by
\begin{align}
\text{\small$\hat{\bar{x}}_{k+1|k}$} &\text{\small$= \mathcal{A}_k\hat{\bar{x}}_{k|k} + \mathcal{B}_ku_k,$} \\
\text{\small$\hat{\bar{x}}_{k|k}$} &\text{\small$= (I-\mathcal{K}_k\mathcal{C}_k)\hat{\bar{x}}_{k|k-1} + \mathcal{K}_k\bar{y}_k - \mathcal{K}_k\mathcal{G}_kh(\hat{x}_{k|k-1}),$} \\
\text{\small$\mathcal{K}_k$} &\text{\small$= \mathcal{P}_{k|k-1}\Phi_k^T(\Phi_k\mathcal{P}_{k|k-1}\Phi_k^T+\mathcal{R})^{-1},$} \\
\text{\small$\mathcal{P}_{k+1|k}$} &\text{\small$= \mathcal{A}_k(I-\mathcal{K}_k\Phi_k)\mathcal{P}_{k|k-1}\mathcal{A}_k^T + \mathcal{Q},$} \\
\text{\small$\mathcal{G}_k$} &\text{\small$\triangleq\begin{bmatrix}G_k\\0\end{bmatrix}, \quad \Phi_k\triangleq\mathcal{C}_k+\begin{bmatrix}0&G_k\frac{\partial h(x_k)}{\partial x_k}\big|_{\hat{x}_{k|k-1}}\\0&0\end{bmatrix}.$} \nonumber
\end{align}
The residue is then $\bar{z}_k = \bar{y}_k - \mathcal{C}_k\hat{\bar{x}}_{k|k-1} - \begin{bmatrix}G_kh(\hat{x}_{k|k-1});\text{ }0\end{bmatrix}$, and the detector is the same as the $\chi^2$ detector in \eqref{Detector} except that $\mathcal{C}_i$ is replaced by $\Phi_i$.

\section{Hybrid Moving Target Defense}
We first consider the hybrid moving target defense, where we perform active detection by changing the parameters of the plant itself in a discrete fashion. This technique will aid not only in the detection of malicious adversaries but 
will also prevent unidentifiable attacks by limiting the adversary's knowledge of the system. To begin, we let $y(x_0,D^ad_k^a,k)$ be the output $y_k$ due to the initial state $x_0$ and the sequence of attacks $\{D^ad_0^a,\cdots,D^ad_k^a\}$, and we let $y_k^s$ represent the $s$th entry of $y_k$ where we have dropped the superscript $a$ (denoting the measurements received by the system operator) for notational simplicity.
\begin{definition}
A nonzero attack on sensor $s$ is unambiguously identifiable at time $t$ if there is no $x_0^* \in \mathbb{R}^n$ satisfying $y_k^s = y^s(x_0^*, 0, k)$ for $0 \le k \le t$. An attack on sensor $s$ is unambiguously identifiable  if it is unambiguously identifiable for all $t$.
\end{definition}
The notion of unambiguous identifiability characterizes when the defender can be certain that sensor $s$ is faulty or under attack. This scenario occurs only if there exists no initial state which produces the output sequence at $y^s$. We seek to design a system that forces the attacker to generate unambiguously identifiable attacks on all targeted sensors, allowing the defender to identify these misbehaving sensors.

We consider the hybrid moving target dynamics as given in (\ref{HybridSystemDynamics}) from the adversary's perspective, where the adversary performs an attack on an ordered set of sensors $L = \{s_1, \cdots, s_{|L|}\}$ using additive inputs $d_k^a \in \mathbb{R}^{|L|}$ such that $y_k=C_kx_k+D^ad_k^a+v_k$. Without loss of generality, we assume that an attack starts at time $k = 0$. Here, $D^a \in \mathbb{R}^{m \times |L|}$ is defined as $D_{uv}^a (L)= \mathbb{I}_{u = s_i, v = i}$ where $\mathbb{I}$ is the indicator function and $(u,v)$ are the indices of an element of $D^a$. Implicitly, we assume that the set of sensors which the adversary targets is constant due to (ideally) the inherent difficulty of hijacking sensors. In an integrity attack, the adversary seeks to adversely affect the physical system by preventing proper feedback.

Consequently, it is important for the defender to identify trusted sensor nodes. Estimation and control algorithms can then be tuned to ignore attacked nodes. We assume that the defender knows the system dynamics $A_k,B_k,C_k$ as well as the input and output histories given by $u_{0:k-1}$ and $y_{0:k}$ but is unaware of the set $L$ and the initial state $x_0$. In addition, we assume that the adversary is limited to sensor attacks. That is, unlike the attack models considered in the extended moving target defense and nonlinear moving target defense, no integrity attacks will be performed on actuators. Hence the problem of identifying malicious nodes is independent of the control input, allowing us to disregard the control input and let $B_k$ be constant. In the deterministic case, the dynamics are then given by
\begin{equation} \small
x_{k+1} = A_k x_k, \quad y_k = C_k x_k + D^a d_k^a. \label{eq:timevarying}
\end{equation}

\begin{remark}
In the deterministic case, we explore attacks where the defender has no knowledge of the initial state. While this is certainly not realistic, the attack vectors developed in this scenario can still remain stealthy in a practical stochastic setting if the adversary carefully ensures that his or her initial attack inputs remain hidden by the noise of the system.
\end{remark}

We now characterize attacks which are not unambiguously identifiable. For notational simplicity let the $s$th row of $C_k$ and $D^a$ be denoted as $C_k^s$ and $D^s$, respectively.
\begin{theorem}
An attack on sensor $s$ in \eqref{eq:timevarying}  is not unambiguously identifiable at time $t$ if and only if there exists an $x_0^*$ such that $D^{s} d_k^a = C_k^s (\prod_{j=0}^{k-1} A_{k-1-j})x_0^*$ for all time $0 \le k \le t$ and $C_k^s (\prod_{j=0}^{k-1} A_{k-1-j})x_0^* \neq 0$ for some time $0 \le k \le t$. \label{MTImage}
\end{theorem}
\begin{proof}
The proof is given in \cite{weerakkody2016moving}.
\end{proof}
Changing the system matrices as a function of time allows the system to act like a moving target. Even if an attacker is aware of the existing configurations $\Gamma$ of the system, he or she will likely be forced to generate unambiguously identifiable attacks since he or she is not aware of the sequence of system matrices.

\subsection{System Design for Deterministic Identification}
We now consider criteria that can allow a defender to design an effective set $\Gamma$. We assume that the adversary knows $\Gamma$, the sequence of attack inputs $D^ad_{0:k}^a$, and the probability distribution of the sequence of system matrices $A_k$ and $C_k$ but does not know the input sequence $u_{0:k-1}$ or the output sequence $y_{0:k}$. Given this knowledge, an adversary can guess the sequence of system matrices and if correct can generate attacks that are not unambiguously identifiable.

We would like to consider systems where $A_k$ and $C_{k}$ remain constant for multiple time steps due to the system's inertia. For now, we assume the pair $(A_{k},C_{k}) \subset \Gamma$ is \emph{constant}. An adversary can use his or her knowledge of $\Gamma$ to guess a pair  $(A_{k},C_{k}) \in \Gamma$ and generate unidentifiable attack inputs. We next determine when an attacker is able to guess an incorrect pair and avoid generating an unambiguously identifiable attack. 

\begin{theorem}
\label{UnambiguousAttackEigenvalues}
 Suppose $(A,C) = (A(1),C(1))$ and an adversary generates a nonzero attack input on sensor $s$ using $(A(2),C(2))$ by inserting attacks along the image of $\mathcal{O}_{t,2}^{{s}}$ where $\mathcal{O}_{t,j}^{{s}}\triangleq\text{\small$\begin{bmatrix}(C^s(j))^T&(C^s(j)A(j))^T\text{ }\cdots\text{ }(C^s(j)A(j)^{t-1})^T\end{bmatrix}^T$}$. Let $\Lambda^j = \{\lambda_1^j, \cdots, \lambda_{q_j}^j\}$ be the set of distinct eigenvalues associated with $A(j)$. Let $\text{\small$\left\{v_{1,1}^{\lambda_i^j},\cdots,v_{1,r_{ij}(1)}^{\lambda_i^j},\cdots,v_{\ell(\lambda_i^j),1}^{\lambda_i^j},\cdots,v_{\ell(\lambda_i^j),r_{ij}(\ell(\lambda_i^j))}^{\lambda_i^j} \right\}$}$ be a maximal set of linearly independent (generalized) eigenvectors associated with eigenvalue $\lambda_i^j$ with $\ell(\lambda_i^j)$ Jordan blocks satisfying
\begin{equation} \small
A(j) v_{\tau,1}^{\lambda_i^j} = \lambda_i^j v_{\tau,1}^{\lambda_i^j}, \quad A(j)v_{\tau,k+1}^{\lambda_i^j} = \lambda_i^j v_{\tau,k+1}^{\lambda_i^j} + v_{\tau,k}^{\lambda_i^j}.
\end{equation}
Let $r_{ij}^M = \underset{t}{\max}~r_{ij}(t)$, and define $V_{s,k}^{\lambda_i^j} \in \mathbb{C}^{r_{ij}^M \times r_{ij}(k)}$ as
\begin{equation*} \small
V_{s,k}^{\lambda_i^j} \triangleq
\begin{bmatrix}
C^s(j) v_{k,1}^{\lambda_i^j} & \cdots & C^s(j) v_{k,r_{ij}(k)}^{\lambda_i^j} \\
0 & \ddots & \vdots \\
0 & \cdots & C^s(j) v_{k,1}^{\lambda_i^j} \\
0 & \cdots & 0
\end{bmatrix}.
\end{equation*}
There exists an attack on sensor $s$ which is not unambiguously identifiable for all time if and only if there exists $\lambda_{i_1}^1 \in \Lambda^1$ and $\lambda_{i_2}^2 \in \Lambda^2$ which satisfy $\lambda_{i_1}^1 = \lambda_{i_2}^2$ and
\begin{equation*} \small
\begin{split}
\textup{Null}\begin{pmatrix}\mathcal{V}_s^{\lambda_{i_1}^1}&\mathcal{V}_s^{\lambda_{i_2}^2}\end{pmatrix} &> \textup{Null}\begin{pmatrix}\mathcal{V}_s^{\lambda_{i_1}^1}\end{pmatrix} + \textup{Null}\begin{pmatrix}\mathcal{V}_s^{\lambda_{i_2}^2}\end{pmatrix}, \\
\mathcal{V}_s^{\lambda_{i_j}^j} \triangleq& \begin{pmatrix}V_{s,1}^{\lambda_{i_j}^j}\text{ }\cdots\text{ }V_{s,\ell(\lambda_{i_j}^j)}^{\lambda_{i_j}^j}\\0_{r_j^M\times\sum_{t=1}^{\ell(\lambda_{i_j}^j)}r_{i_jj}(t)}\end{pmatrix},
\end{split}
\end{equation*}
with $r_1^M=r_{i_22}^M-r_{i_11}^M$ if $r_{i_11}^M<r_{i_22}^M$, $r_2^M=r_{i_11}^M-r_{i_22}^M$ if $r_{i_11}^M>r_{i_22}^M$, and $r_j^M=0$ otherwise.

Otherwise the attack can be detected in time $t \le 2n-1$. \label{theorem:det}
\end{theorem}
\begin{proof}
The proof is given in \cite{weerakkody2016movingarxiv}.
\end{proof}
Roughly speaking, given enough observations, the output at sensor $s$ for a time-invariant system will be dominated by the observable mode(s) that have the largest eigenvalue. Thus, if the eigenvalues between two system matrices are distinct, we are able to distinguish the resulting outputs. Theorem \ref{UnambiguousAttackEigenvalues} gives the defender an efficient way to determine if an attacker can guess $\Gamma$ incorrectly and still remain unidentified when the system matrices are kept constant for at least $2n$ time steps. It also prescribes a means to perform perfect identification.\\
\textbf{Design Recommendations}
\begin{enumerate}
\item  For all pairs $i \neq j \in \{1, \cdots, l\}$, $\Lambda^i \cap \Lambda^j = \emptyset$.
\item  The system matrices $(A_k,C_k)$ are periodically changed after every $\kappa \ge 2n$ time steps.
\item  Let $\{l_k\}$ be a sequence where $l_k \in \{1, \cdots, l\}$. Let $\varrho_k$ denote the indices of a subsequence. $\mbox{Pr}((A_{\varrho_k},C_{\varrho_k}) = (A(l_k),C(l_k)), \text{ }\forall k) = 0$.
\item  The pair $(A(i),C(i))$ is observable $\forall i \in \{1, \cdots, l\}$.
\item  For all $i \in \{1, \cdots, l\}$, $0 \notin \Lambda^i$.  
\end{enumerate}

\begin{corollary}
Assume a defender follows the design recommendations. Suppose sensor $s$ is attacked and there is no $t^*$ such that $D^{s} d_k^a = 0$ for all $k \ge t^*$. Then the sensor attack will be unambiguously identifiable with probability 1. \label{cor:design}
\end{corollary} 
\begin{proof}
The proof is given in \cite{weerakkody2016moving}.
\end{proof}
As a result, an attacker who persistently biases a sensor will be perfectly identified. Note that recommendation 3 can be achieved with an IID assumption or an aperiodic and irreducible Markov chain. The last 2 recommendations are not needed for this result but are justified in the next subsection when we consider stochastic systems.

\begin{remark}
Keeping the system matrices constant for a long enough period of time appears counter-intuitive for the hybrid moving target. However, the given adversary is not performing system identification and is instead guessing the system matrices. As such, keeping the dynamics constant does not provide useful information for an attacker. Additionally, keeping the matrices constant long enough gives the defender the information he or she needs to distinguish between the different hybrid states. Similar to the problem of observability, the problem of identification involves a rank deficient matrix until enough measurements have been gathered.
\end{remark}

\subsection{False Data Injection Detection}
We now examine the effectiveness of the hybrid moving target defense for detection in the case of a stochastic system where the dynamics are given by
\begin{equation} \small
x_{k+1} = A_k x_k + w_k, \quad y_k = C_k x_k + D^a d_k^a + v_k. \label{eq:timevaryingstoc}
\end{equation}

The information and goals of the adversary and the defender remain unchanged except that both the adversary and the defender are aware of the noise statistics and the defender knows the distribution of the initial state $x_0\sim\mathcal{N}(\hat{x}_{0|-1},P_{0|-1})$. To characterize detection performance, we consider the additive bias $\Delta z_k$ the adversary injects on the normalized residues due to his or her sensor attacks. The normalized residue is the normalized difference between the observed measurement and its expected value, which is slightly different than the unnormalized residue defined in section II that is used throughout the rest of the article. The bias $\Delta e_k\triangleq x_k-\hat{x}_{k|k}$ on the a posteriori state estimation error and the bias $\Delta z_k$ on the normalized residues are given by 
\begin{align}
\text{\small$\Delta e_k$} &\text{\small$= (A_{k-1}-K_kC_kA_{k-1}) \Delta e_{k-1} - K_k D^a d_k^a,$} \\
\text{\small$\Delta z_k$} &\text{\small$= (C_k P_{k|k-1} C_k^T + R)^{-\frac{1}{2}}\left( C_k A_{k-1} \Delta e_{k-1} + D^a d_k^a \right),$}
\end{align}
with $\Delta e_{-1} = 0$. A residue detector such as the $\chi^2$ detector will recognize large residues and mark them as belonging to an attack. We now show that an adversary is restricted in the bias he or she can inject on the state estimation error without significantly biasing the residues and incurring detection.
\begin{theorem} 
Suppose a defender uses a hybrid moving target defense leveraging the design recommendations listed above. Then  $\limsup_{k \rightarrow \infty} \| \Delta e_k \| = \infty \implies  \limsup_{k \rightarrow \infty} \| \Delta z_k \| = \infty$ with probability 1. \label{theorem:fdi}
\end{theorem}
\begin{proof}
The proof is given in \cite{weerakkody2016moving}.
\end{proof}
Thus the attacker is able to destabilize the estimation error only by destabilizing the residues. As such, there is a point where an attacker is unable to introduce additional bias to the estimation error without revealing his or her presence due to his or her effect on the measurement residues.

\subsection{Resilient Estimation and Identification}
While the hybrid moving target defense guarantees we can detect unbounded false data injection attacks, we want to also identify specific malicious sensors as in the deterministic case. To do so, we present a resilient estimator that fuses state estimates generated by individual sensors since previous results \cite{sun2004multi,gan2001comparison} suggest such an estimator has better fault tolerance. This is desirable since we are attempting to force a normally stealthy adversary to generate faults. We will show that an attacker can destabilize this estimator only if the culprit sensors can be identified. In particular, we will show that the estimation error will become unbounded only if the bias on a sensor residue is also unbounded.

To begin, we assume that for each sensor $s$, $\text{NS}(\mathcal{O}_{n,1}^{s}) = \cdots = \text{NS}(\mathcal{O}_{n,l}^{s})$, where $\text{NS}(A)$ denotes the null space of $A$. Such a condition is realistic since it implies that changing the system dynamics does not affect what portion of the state the sensor itself can observe. Using a Kalman decomposition for each sensor $s$, there exists a state transformation $\mathcal{T}_s \triangleq \begin{bmatrix} \mathcal{T}_s^{uo} & \mathcal{T}_s^o \end{bmatrix}$ such that $\begin{bmatrix} \mathcal{T}_s^{uo} & \mathcal{T}_s^o \end{bmatrix} \text{\small$\begin{bmatrix} \zeta_{k,s}^{{uo}^T} & \zeta_{k,s}^T \end{bmatrix}^T$} = x_k$ and $\begin{bmatrix} \mathcal{T}_s^{uo} & \mathcal{T}_s^o \end{bmatrix} \text{\small$\begin{bmatrix} \psi_{k,s}^{{uo}^T} & \psi_{k,s}^T \end{bmatrix}^T$} = w_k$. Here, the columns of $\mathcal{T}_s^{uo}$ are a basis for $\text{NS}(\mathcal{O}_{n,1}^{s})$ while the columns of $\mathcal{T}_s^o$ should be chosen so the resulting $\mathcal{T}_s$ is invertible. Using the same transform $\mathcal{T}_s$ for each mode in $\Gamma$, there exists a  $\Gamma^s = \{(C_s(1),A_s(1)),\cdots,(C_s(l),A_s(l))\}$ corresponding to $\Gamma$ such that
\begin{equation} \small
\zeta_{k+1,s} = A_{k,s} \zeta_{k,s} + \psi_{k,s}, \quad y_k^s = C_{k,s} \zeta_{k,s} + v_k^s,
\end{equation}
where each pair $(A_{k,s},C_{k,s})$ is observable and belongs to $\Gamma^s$.

A Kalman filter with bounded covariance (see proof of Theorem \ref{theorem:fdi}) can be constructed to estimate $\zeta_{k,s}$ given $y_{0:k}^s$. From the definition of the Kalman filter, we have
\begin{align}
\label{eq:kalmanprior}
\text{\small$\hat{\zeta}_{k+1|k,s}$} &\text{\small$= A_{k,s} \hat{\zeta}_{k|k,s},$} \\
\label{eq:kalmanposterior}
\text{\small$\hat{\zeta}_{k|k,s}$} &\text{\small$= (I - K_{k,s} C_{k,s}) \hat{\zeta}_{k|k-1,s} + K_{k,s} y_k^s,$} \\
\label{eq:kalmangain}
\text{\small$K_{k,s}$} &\text{\small$= P_{k|k-1}^{s,s} C_{k,s}^T (C_{k,s}  P_{k|k-1}^{s,s} C_{k,s}^T + R_{s,s})^{-1},$} \\
\label{eq:kalmanpriorcov}
\text{\small$P_{k+1|k}^{s_1,s_2}$} &\text{\small$= A_{k,s_1} P_{k|k}^{s_1,s_2} A_{k,s_2}^T + Q_{s_1,s_2},$} \\
\label{eq:kalmanposteriorcov}
\text{\small$P_{k|k}^{s_1,s_2}$} &\text{\small$= (I - K_{k,s_1} C_{k,s_1}) P_{k|k-1}^{s_1,s_2} (I - K_{k,s_2} C_{k,s_2})^T$} \\
&\text{\small$\quad+ K_{k,s_1}R_{s_1,s_2}K_{k,s_2}^T,$} \nonumber \\
\label{eq:normalizedresidue}
\text{\small$z_{k,s}$} &\text{\small$= (C_{k,s}  P_{k|k-1}^{s,s} C_{k,s}^T + R_{s,s})^{-\frac{1}{2}}(y_k^s - C_{k,s}\hat{\zeta}_{k|k-1,s}),$}
\end{align}
where $\hat{\zeta}_{k|k-1,s}$ and $\hat{\zeta}_{k|k,s}$ are the a priori and a posteriori state estimates of $\zeta_{k,s}$, $P_{k|k-1}^{s_1,s_2} \triangleq \mathbb{E}[e_{k|k-1,s_1}e_{k|k-1,s_2}^T]$ and $P_{k|k}^{s_1,s_2} \triangleq \mathbb{E}[e_{k|k,s_1}e_{k|k,s_2}^T]$ are the a priori and a posteriori error covariance matrices with $e_{k|k-1,s}\triangleq\zeta_{k,s}-\hat{\zeta}_{k|k-1,s}$ and $e_{k|k,s}\triangleq\zeta_{k,s}-\hat{\zeta}_{k|k,s}$, $K_{k,s}$ is the Kalman gain, $Q_{s_1,s_2}\triangleq\mathbb{E}[\psi_{k,s_1}\psi_{k,s_2}^T]$, $R_{i,j}$ is the $(i,j)$th entry of $R$, and $z_{k,s}$ is the normalized residue. Note that \eqref{eq:kalmanpriorcov} and \eqref{eq:kalmanposteriorcov} hold for $s_1 = s_2$.

We would like to use the individual state estimates $\hat{\zeta}_{k|k,s}$ associated with each sensor $s$ to obtain an overall state estimate of $x_k$. To do this, first define $x_{k,s}^o$ as $x_{k,s}^o \triangleq \mathcal{T}_s^{o}\hat{\zeta}_{k|k,s} + \alpha_{k,s}$, where $\alpha_{k,s}$ is an IID sequence of Gaussian random variables with $\alpha_{k,s} \sim \mathcal{N}(0,\sigma I)$ for some small $\sigma > 0$. Moreover $\{\alpha_{k,s_1}\}$ and $\{\alpha_{k,s_2}\}$ are independent sequences. $\alpha_{k,s}$ is a mathematical artifact introduced so the subsequent estimator has a simplified closed form and can be easily removed or mitigated by letting $\sigma$ tend to 0. From here we obtain $\mathbf{\hat{y}_k} = W \mathbf{x_k} + \mathbf{\alpha_k}$, where $\mathbf{\hat{y}_k} \triangleq \text{\small$\begin{bmatrix}  x_{k,1}^{o^T} \text{ } \cdots \text{ } x_{k,m}^{o^T} \end{bmatrix}^T$}$, $\mathbf{x_k} \triangleq \text{\small$\begin{bmatrix} \zeta_{k,1}^{uo^T} \text{ } \cdots \text{ } \zeta_{k,m}^{uo^T} & x_k^T \end{bmatrix}^T$}$,
\begin{equation*} \small
\medmuskip=3.62mu
\thinmuskip=3.62mu
\thickmuskip=3.62mu
\mathbf{\alpha_k} \triangleq 
\begin{bmatrix} -\mathcal{T}_1^{o} e_{k|k,1} + \alpha_{k,1} \\ \vdots \\ -\mathcal{T}_m^{o} e_{k|k,m} + \alpha_{k,m} \end{bmatrix},
W \triangleq \begin{bmatrix} -\mathcal{T}_1^{uo} & \cdots & 0 & I \\
\vdots & \ddots & \vdots & \vdots \\
0 & \cdots & -\mathcal{T}_m^{uo} & I \end{bmatrix}.
\end{equation*}
It can be seen that $\mathbf{\alpha_k}$ is normally distributed so that $\mathbf{\alpha_k} \sim \mathcal{N}({0},\Upsilon)$, where $\Upsilon \succ 0$ consists of $m \times m$ blocks with the $(i,j)$th block given by $\mathcal{T}_i^o P_{k|k}^{i,j} \mathcal{T}_j^{oT}+ \delta_{ij}\sigma I$. Here $\delta_{ij}$ is the Kronecker delta. The minimum variance unbiased estimate (MVUB) \cite{Scharf1990} of $\mathbf{x}_k$ given $\mathbf{\hat{y}_k}$ is given by 
\begin{equation} \small
\label{eq:optest}
\mathbf{\hat{x}}_{k} = (W^T \Upsilon^{-1} W)^{-1} W^T \Upsilon^{-1} \mathbf{\hat{y}_k}.
\end{equation}
The last $n$ entries of $\mathbf{\hat{x}_{k}}$, denoted as $\hat{x}_k^*$, constitute an MVUB estimate of $x_k$ given the set of sensor estimates $\mathbf{\hat{{y}}_k}$. We next show that the proposed estimator of $x_k$ has bounded covariance.
\begin{theorem}
Consider the estimator of $x_k$ defined by \cref{eq:kalmanprior,eq:kalmanposterior,eq:kalmangain,eq:kalmanpriorcov,eq:kalmanposteriorcov,eq:normalizedresidue,eq:optest}. The estimator has bounded covariance. \label{theorem:boundedcov}
\end{theorem}
\begin{proof}
The proof is given in \cite{weerakkody2016moving}.
\end{proof}

We lastly demonstrate that the proposed estimator is sensitive to biases $\Delta z_{k,s}$ in individual normalized residues, specifically showing that an infinite bias introduced into the estimator implies that the residues are also infinite. Defining $e_k^* \triangleq x_k - \hat{x}_k^*$ and letting $\Delta e_k^*$ represent the bias on $e_k^*$ due to the adversary's inputs, we have the following result.
\begin{theorem}
Consider the estimator of $x_k$ defined by \cref{eq:kalmanprior,eq:kalmanposterior,eq:kalmangain,eq:kalmanpriorcov,eq:kalmanposteriorcov,eq:normalizedresidue,eq:optest}. Then, with probability 1, $\limsup_{k \rightarrow \infty} \| \Delta e_k^* \| = \infty \implies \limsup_{k \rightarrow \infty} \| \Delta z_{k,i} \| = \infty$ for some $i \in \{1, \cdots, m\}$.
\end{theorem} 
\begin{proof}
The proof is given in \cite{weerakkody2016moving}.
\end{proof}
While the proposed estimator does not guarantee each malicious sensor will be identified, it does guarantee that the defender will be able to identify and remove sensors whose attacks cause unbounded bias in the estimation error simply by analyzing each sensor's measurements individually. This is due to the fact that the bias on the residues of such sensors will grow unbounded, which can be easily detected by a $\chi^2$ detector. As a result, we propose the following detector to identify malicious behavior for each individual sensor $s$
\begin{equation} \small
\label{chi}
g_{k,s}(z_{k-T+1:k,s}) = \sum_{j = k-T+1}^k z_{j,s}^2 \overset{{\mathcal{H}_1^s}}{\underset{\mathcal{H}_0^s}{ \gtrless}} \tau_k^i.
\end{equation}
Here $g_{k,s}$ is the detection statistic for sensor $s$, $\tau_k^i$ represents the threshold of the detector, $i\in\{1,\cdots,m\}$, and $\mathcal{H}_1^s$ and $\mathcal{H}_0^s$ are the hypotheses that sensor $s$ is malfunctioning or is working normally, respectively. A sensor $s$ which repeatedly fails detection can be removed from consideration when obtaining a state estimate and the proposed fusion-based estimation scheme can be adjusted accordingly.

\begin{remark}
Consider a standard LTI system whose dynamics are known to an attacker. From \cite{fawzi2014secure}, there exists an estimator which can recover the system state with up to $q$ sensor attacks if and only if the system is $2q$ sparse observable (that is observable even if any $2q$ sensors are removed). The hybrid moving target allows us to perform perfect (stable) state estimation in the deterministic (stochastic) scenario if the system is merely $q$ sparse observable. In the deterministic case, we simply identify the $q$ attacked sensors, remove them, and use the healthy sensors to recover the initial state. In the stochastic case, we can identify $q$ sensors that cause any destabilizing estimation errors, remove them, and then use the proposed fusion-based state estimator to obtain a stable estimate.
\end{remark}

\section{Extended Moving Target Defense}
Instead of varying the system matrices directly, the extended moving target defense introduces an auxiliary system whose sensor measurements reveal any biases an adversary exerts on the nominal control system. As such, we seek to design the auxiliary system in such a way as to maximize the probability of detection when the system is under attack. Specifically, we would like to design the parameters that generate $\bar{A}_k$, $\tilde{B}_k$, and $\bar{C}_k$ to maximize detection performance. Because a joint maximization over $\bar{A}_k$, $\tilde{B}_k$, and $\bar{C}_k$ becomes infeasible for $\bar{A}_k$ and $\bar{C}_k$, we recognize that detection performance is a direct function of accurate state estimation. Consequently, we design the parameters that generate $\bar{A}_k$ and $\bar{C}_k$ to maximize estimation performance while designing the parameters that generate $\tilde{B}_k$ to maximize detection performance.

\begin{remark}
For notational simplicity, we will assume that $\bar{A}_k$, $\tilde{B}_k$, and $\bar{C}_k$ are not sparse. However, the analysis presented in this section can easily be extended to designs where $\bar{A}_k$, $\tilde{B}_k$, and $\bar{C}_k$ are sparse matrices.
\end{remark}

We consider a general set of additive integrity attacks as modeled in (\ref{IntegrityAttackStatesSensors}) for the nominal system where $B^a=B$ (all actuators can be corrupted) and $D^a=I$ (all sensors can be modified). An adversary with these capabilities and knowledge of the nominal system dynamics can arbitrarily and stealthily perturb the nominal system using a covert attack \cite{smith2015covert}. This set of additive integrity attacks can be written as
\begin{equation} \small
\label{AttackStateSensorDynamics}
\bar{x}_{k+1}^A = \mathcal{A}_k\bar{x}_k^A + \mathcal{B}_k (u_k+u_k^a) + \bar{w}_k, \quad \bar{y}_k^a = \mathcal{C}_k\bar{x}_k^A + \bar{d}_k^a + \bar{v}_k,
\end{equation}
where $\bar{x}_k^A$ represents the attacked states, $u_k^a$ denotes the attacker's additive bias on the control inputs, $\bar{d}_k^a$ represents the attacker's additive bias on the sensor measurements, and $\bar{y}_k^a$ denotes the biased sensor measurements received by the system operator. Here the auxiliary actuators $\tilde{B}_k$ and coupling matrices $\bar{A}_k$ and $\bar{C}_k$ are generated from the following distributions: $\tilde{B}_k(\text{row }i)\sim\mathcal{N}(\mu_{\tilde{B}},\Sigma_{\tilde{B}_k})\text{ }\forall i$, $\bar{A}_k(\text{row }i)\sim\mathcal{N}(\mu_{\bar{A}},\Sigma_{\bar{A}})\text{ }\forall i$, and $\bar{C}_k(\text{row }i)\sim\mathcal{N}(\mu_{\bar{C}},\Sigma_{\bar{C}})\text{ }\forall i$ with independence between rows over time. We consider a strong adversary who is able to read and modify all of the inputs and outputs so that the design of the parameters generating $\bar{A}_k$, $\tilde{B}_k$, and $\bar{C}_k$ is optimal for even the strongest additive integrity attacks. Given this attack model, we now describe how to design the covariances $\Sigma_{\tilde{B}_k}$, $\Sigma_{\bar{A}}$, and $\Sigma_{\bar{C}}$ of the distributions associated with the auxiliary actuators and the coupling matrices to maximize detection and estimation performance, respectively.

\subsection{Auxiliary Actuators Design}
To design the covariance $\Sigma_{\tilde{B}_k}$ to maximize detection performance, we use the Kullback-Liebler (KL) divergence as a metric for detection performance that, roughly speaking, quantifies the distance between the distribution of the residue under attack and the distribution of the residue under normal operation. We note that any additive integrity attack will result in an additive bias on the residue which can be written as a linear combination of the control input biases $u_{j:k-1}^a$ and sensor measurement biases $\bar{d}_{j+1:k}^a$ exerted by the attacker. Here $j$ denotes the time when the attacker first exerts a bias on the control inputs and $j+1$ represents the time when the attacker first attempts to hide his or her attack by exerting a bias on the sensor measurements. As shown in \cite{acc2017griffioen}, the bias on the residue $\Delta\bar{z}_i$ can be written as
\begin{equation} \small
\underbrace{\begin{bmatrix}
M_{(j,i)}^x & -M_{(j,i)}^y
\end{bmatrix}}_{M_{(j,i)}}
\underbrace{\begin{bmatrix}
u_j^{a^T} \text{ } \cdots \text{ } u_{k-1}^{a^T} &
\bar{d}_{j+1}^{a^T} \text{ } \cdots \text{ } \bar{d}_k^{a^T}
\end{bmatrix}^T}_{\phi_{j:k}},
\end{equation}
where $\phi_{j:k}$ represents a vector containing all of the attacker's biases and $M_{(j,i)}^x$ and $M_{(j,i)}^y$ are given by
\begin{align}
\label{Mxji}
\text{\small$M_{(j,i)}^x$} &\text{\small$\triangleq \begin{bmatrix}
\mathcal{C}_i\mathcal{D}_{(j,i)}\mathcal{B}_j \text{ } \cdots \text{ } \mathcal{C}_i\mathcal{D}_{(i-1,i)}\mathcal{B}_{i-1} & 0_i \text{ } \cdots \text{ } 0_{k-1}
\end{bmatrix},$} \\
\text{\small$M_{(j,i)}^y$} &\text{\small$\triangleq \begin{bmatrix}
\Xi_{(j+1,i)} \text{ } \cdots \text{ } \Xi_{(i-1,i)} & I & 0_{i+1} \text{ } \cdots \text{ } 0_k
\end{bmatrix},$}
\end{align}
with $\mathcal{D}_{(j,i)}\triangleq\prod_{t=j+1}^{i-1}\mathcal{A}_{i+j-t}(I-\mathcal{K}_{i+j-t}\mathcal{C}_{i+j-t})$, $\Xi_{(j,i)}\triangleq\mathcal{C}_i\mathcal{D}_{(j,i)}\mathcal{A}_j\mathcal{K}_j$, $0_t\in\mathbb{R}^{(m+\tilde{m})\times p}$ for $M_{(j,i)}^x$, and $I,0_t\in\mathbb{R}^{(m+\tilde{m})\times(m+\tilde{m})}$ for $M_{(j,i)}^y$.

Under normal operation, it can be shown that the residue follows a normal distribution $f_0(\bar{z}_i) = \mathcal{N}(0,\Sigma_i)$ with zero mean and covariance $\Sigma_i\triangleq\mathcal{C}_i\mathcal{P}_{i|i-1}\mathcal{C}_i^T+\mathcal{R}$. If the defender has no prior information about the attacker's biases $\phi_{j:k}$, the residue under attack will also follow a normal distribution $f_1(\bar{z}_i) = \mathcal{N}(M_{(j,i)}\phi_{j:k},\Sigma_i)$ with a mean equal to the bias on the residue. Because the covariance of $f_0$ and $f_1$ are the same, the KL divergence is symmetric and can be written as
\begin{equation} \small
\begin{split}
D_{KL}&\left(f_1(\bar{z}_{k-T+1:k})||f_0(\bar{z}_{k-T+1:k})\right) = \\
&\quad = \mathbb{E}_{\bar{z}_{k-T+1:k}}\Bigg[\frac{1}{2}\Bigg(\sum_{i=k-T+1}^k-\bar{z}_i^T\Sigma_i^{-1}\bar{z}_i + \\
&\quad\quad +(\bar{z}_i-M_{(j,i)}\phi_{j:k})^T\Sigma_i^{-1}(\bar{z}_i-M_{(j,i)}\phi_{j:k})\Bigg)\Bigg|\mathcal{H}_0\Bigg] \\
&\quad = \frac{1}{2}\phi_{j:k}^T\left(\sum_{i=k-T+1}^kM_{(j,i)}^T\Sigma_i^{-1}M_{(j,i)}\right)\phi_{j:k},
\end{split}
\end{equation}
where the second equality follows from the fact that the residue has zero mean under normal operation.

Maximizing the KL divergence becomes difficult because the attacker biases $\phi_{j:k}$ are unknown to the defender. However, we note that $M_{(j,i)}^T\Sigma_i^{-1}M_{(j,i)}$ is positive semidefinite, allowing us to maximize the expected value of the KL divergence for all possible additive integrity attacks. This is carried out by maximizing a nonnegative constant $\epsilon$ such that the expected value of the KL divergence is greater than a positive semidefinite lower bound $N(\epsilon)$ that is a function of $\epsilon$. Since there are real-world constraints on the variance magnitude of the auxiliary actuators, we constrain the covariance $\Sigma_{\tilde{B}_k}$ with a positive semidefinite upper bound $N_B$. This maximization problem is presented below
\begin{equation} \small
\label{InitialOptimization}
\begin{split}
& \argmax_{\epsilon,\Sigma_{\tilde{B}_k}} \epsilon \quad \text{s.t. } \Sigma_{\tilde{B}_k} \preceq N_B, \\
& \quad \quad \frac{1}{2}\mathbb{E}_{\tilde{B}_{j:k-1}}\left[\sum_{i=k-T+1}^kM_{(j,i)}^T\Sigma_i^{-1}M_{(j,i)}\right] \succeq N(\epsilon).
\end{split}
\end{equation}

As shown in \cite{acc2017griffioen}, we can construct the positive semidefinite lower bound $N(\epsilon)$ to match the block structure of $\sum_{i=k-T+1}^kM_{(j,i)}^T\Sigma_i^{-1}M_{(j,i)}$. The off-diagonal blocks of this structure are not functions of $\Sigma_{\tilde{B}_k}$, allowing the constraint in (\ref{InitialOptimization}) to be simplified to the following series of constraints
\begin{equation} \small
\label{ConstraintV3}
\begin{split}
& \frac{1}{2}\sum_{i=k-T+1}^k\Psi_{(j,i)}^{tt}\succeq \epsilon N_{t}, \quad t = 1,\cdots,i-j
\end{split}.
\end{equation}
Here $N_t\succeq0$, and $\Psi_{(j,i)}^{tt}$ is shown in \cite{acc2017griffioen} to be
\begin{equation} \small
\label{PsiDef}
\begin{split}
& \Psi_{(j,i)}^{tt} = \text{Tr}(\tilde{D}_{(j+t-1,i)}^T\mathcal{C}_i^T\Sigma_i^{-1}\mathcal{C}_i\tilde{D}_{(j+t-1,i)})\Sigma_{\tilde{B}_{j+t-1}} \\
& \quad + \text{Sum}(\tilde{D}_{(j+t-1,i)}^T\mathcal{C}_i^T\Sigma_i^{-1}\mathcal{C}_i\tilde{D}_{(j+t-1,i)})\mu_{\tilde{B}}\mu_{\tilde{B}}^T \\
& \quad + B^T\bar{D}_{(j+t-1,i)}^T\mathcal{C}_i^T\Sigma_i^{-1}\mathcal{C}_i\bar{D}_{(j+t-1,i)}B \\
& \quad + \begin{bmatrix}\mu_{\tilde{B}}&\cdots&\mu_{\tilde{B}}\end{bmatrix}\tilde{D}_{(j+t-1,i)}^T\mathcal{C}_i^T\Sigma_i^{-1}\mathcal{C}_i\bar{D}_{(j+t-1,i)}B \\
& \quad + B^T\bar{D}_{(j+t-1,i)}^T\mathcal{C}_i^T\Sigma_i^{-1}\mathcal{C}_i\tilde{D}_{(j+t-1,i)}\begin{bmatrix}\mu_{\tilde{B}}&\cdots&\mu_{\tilde{B}}\end{bmatrix}^T,
\end{split}
\end{equation}
where $\begin{bmatrix}\tilde{D}_{(j+t-1,i)}&\bar{D}_{(j+t-1,i)}\end{bmatrix}\triangleq\mathcal{D}_{(j+t-1,i)}$ and $\text{Sum}(A)$ represents the sum of all the elements of $A$.

Combining (\ref{ConstraintV3}) and (\ref{PsiDef}) allows the second constraint in (\ref{InitialOptimization}) to be written as a set of positive semidefinite constraints. We choose $j=k-T$ so that the KL divergence is maximized over the time window $T$ of the chi-squared detector. This allows the optimization problem in (\ref{InitialOptimization}) to be written as
\begin{equation} \small
\label{BBkOptimization}
\begin{split}
& \argmax_{\epsilon,\Sigma_{\tilde{B}_k}}\epsilon \quad \text{s.t. } \Sigma_{\tilde{B}_k} \preceq N_B, \\
& \quad \frac{1}{2}\sum_{i=k-T+1}^{k-t}\text{Tr}(\tilde{D}_{(k-T+t,i)}^T\mathcal{C}_i^T\Sigma_i^{-1}\mathcal{C}_i\tilde{D}_{(k-T+t,i)})\Sigma_{\tilde{B}_{k-T+t}} \\
& \quad + \text{Sum}(\tilde{D}_{(k-T+t,i)}^T\mathcal{C}_i^T\Sigma_i^{-1}\mathcal{C}_i\tilde{D}_{(k-T+t,i)})\mu_{\tilde{B}}\mu_{\tilde{B}}^T \\
& \quad + \begin{bmatrix}\mu_{\tilde{B}}\text{ }\cdots\text{ }\mu_{\tilde{B}}\end{bmatrix}\tilde{D}_{(k-T+t,i)}^T\mathcal{C}_i^T\Sigma_i^{-1}\mathcal{C}_i\bar{D}_{(k-T+t,i)}B \\
& \quad + B^T\bar{D}_{(k-T+t,i)}^T\mathcal{C}_i^T\Sigma_i^{-1}\mathcal{C}_i\tilde{D}_{(k-T+t,i)}\begin{bmatrix}\mu_{\tilde{B}}\text{ }\cdots\text{ }\mu_{\tilde{B}}\end{bmatrix}^T \\
& \quad + B^T\bar{D}_{(k-T+t,i)}^T\mathcal{C}_i^T\Sigma_i^{-1}\mathcal{C}_i\bar{D}_{(k-T+t,i)}B \succeq \epsilon N_{t+1},
\end{split}
\end{equation}
with $t = 0,\cdots,T-1$. If $\Sigma_{\tilde{B}_k}$ is time-invariant, this becomes a semidefinite program with unique solutions for $\epsilon$ and $\Sigma_{\tilde{B}_k}$. However, to make this optimization problem solvable when $\Sigma_{\tilde{B}_k}$ is time-varying, we reformulate the optimization by considering a global optimization problem of \eqref{BBkOptimization} over all time steps $k$. We add a subscript to $\epsilon$ in the second constraint of \eqref{BBkOptimization} so that it is now $\epsilon_{k-T+t}$, distinguishing it from $\epsilon$'s at other time steps. Furthermore, we modify the objective function to be $\sum_{k=0}^\infty\epsilon_k$. This formulation allows the global optimization problem to be divided into optimization problems for each time step $k$, resulting in
\begin{equation} \small
\label{BBkOptimization2}
\begin{split}
\argmax_{\epsilon_k,\Sigma_{\tilde{B}_k}}\epsilon_k \quad &\text{s.t. } \Sigma_{\tilde{B}_k} \preceq N_B, \\
& \frac{1}{2}\sum_{i=k+1-t}^{k+T-2t}\text{Tr}(\tilde{D}_{(k,i)}^T\mathcal{C}_i^T\Sigma_i^{-1}\mathcal{C}_i\tilde{D}_{(k,i)})\Sigma_{\tilde{B}_k} \\
& + \text{Sum}(\tilde{D}_{(k,i)}^T\mathcal{C}_i^T\Sigma_i^{-1}\mathcal{C}_i\tilde{D}_{(k,i)})\mu_{\tilde{B}}\mu_{\tilde{B}}^T \\
& + \begin{bmatrix}\mu_{\tilde{B}}\text{ }\cdots\text{ }\mu_{\tilde{B}}\end{bmatrix}\tilde{D}_{(k,i)}^T\mathcal{C}_i^T\Sigma_i^{-1}\mathcal{C}_i\bar{D}_{(k,i)}B \\
& + B^T\bar{D}_{(k,i)}^T\mathcal{C}_i^T\Sigma_i^{-1}\mathcal{C}_i\tilde{D}_{(k,i)}\begin{bmatrix}\mu_{\tilde{B}}\text{ }\cdots\text{ }\mu_{\tilde{B}}\end{bmatrix}^T \\
& + B^T\bar{D}_{(k,i)}^T\mathcal{C}_i^T\Sigma_i^{-1}\mathcal{C}_i\bar{D}_{(k,i)}B \succeq \epsilon_k N_{t+1},
\end{split}
\end{equation}
which is a semidefinite program with $t = 0,\cdots,T-1$ and unique solutions for $\epsilon_k$ and $\Sigma_{\tilde{B}_k}$. As a result, this optimization problem provides a method for designing $\Sigma_{\tilde{B}_k}$ so that the expected value of the KL divergence is maximized for the set of all possible additive integrity attacks described by (\ref{AttackStateSensorDynamics}).

\begin{remark}
\label{JointDesign}
One could alternatively consider formulating an optimization problem around the KL divergence as a function of both the mean $\mu_{\tilde{B}}$ and the covariance $\Sigma_{\tilde{B}_k}$. However, such formulations will be non-convex. Addressing the joint design of $\mu_{\tilde{B}}$ and $\Sigma_{\tilde{B}_k}$ using non-convex techniques along with alternative formulations of \eqref{InitialOptimization} are left for future work. This remark also applies to the joint design of $\mu_{\bar{A}}$, $\mu_{\bar{C}}$, $\Sigma_{\bar{A}}$, and $\Sigma_{\bar{C}}$ in \eqref{InitialBarOptimization} and $\mu_{G}$ and $\Sigma_G$ in \eqref{InitialNonlinearOptimization}.
\end{remark}

\subsection{Coupling Matrices Design}
We now focus our attention on the design of the covariances $\Sigma_{\bar{A}}$ and $\Sigma_{\bar{C}}$ that generate the coupling matrices $\bar{A}_k$ and $\bar{C}_k$. To maximize estimation performance, we seek to design $\Sigma_{\bar{A}}$ and $\Sigma_{\bar{C}}$ to maximize the amount of information the defender receives about the attacked states $x_k^A$ through the biased auxiliary sensor measurements $\tilde{y}_k^a$. Accurate estimation of the attacked state will enable a system operator to better distinguish between true and falsified measurements. Since the accuracy of the state estimate depends on the amount of information the sensor measurements carry about the state, maximizing this amount of information should increase detection performance.

\begin{remark}
If the system is operating normally, this design will maximize the amount of information the system operator receives about the unaltered states $x_k$. Consequently, this design will increase estimation performance regardless of whether or not the system is under attack.
\end{remark}

We consider the amount of information all the biased auxiliary sensor measurements $\tilde{y}_{0:k}^a$ contain about all the attacked states $x_{0:k}^A$. As shown in \cite{acc2017griffioen}, we can represent all the biased auxiliary sensor measurements $\tilde{y}_{0:k}^a$ as
\begin{equation} \small
\tilde{y}_{0:k}^a = (H_A+H_C)x_{0:k}^A + H_Bu_{0:k}^A + \tilde{d}_{0:k}^a + H_W\tilde{w}_{0:k-1} + \tilde{v}_{0:k},
\end{equation}
where $\tilde{y}_{0:k}^a \triangleq \text{\small$\begin{bmatrix} \tilde{y}_0^{a^T} \text{ } \cdots \text{ } \tilde{y}_k^{a^T} \end{bmatrix}^T$}$, $x_{0:k}^A \triangleq \text{\small$\begin{bmatrix} x_0^{A^T} \text{ } \cdots \text{ } x_k^{A^T} \end{bmatrix}^T$}$, $u_{0:k}^A \triangleq \text{\small$\begin{bmatrix} u_0^{A^T} \text{ } \cdots \text{ } u_k^{A^T} \end{bmatrix}^T$}$, $u_k^A\triangleq u_k+u_k^a$, $\tilde{d}_{0:k}^a \triangleq \text{\small$\begin{bmatrix} \tilde{d}_0^{a^T} \text{ } \cdots \text{ } \tilde{d}_k^{a^T} \end{bmatrix}^T$}$, $\tilde{w}_{0:k-1} \triangleq \text{\small$\begin{bmatrix} \tilde{x}_0^T & \tilde{w}_0^T \text{ } \cdots \text{ } \tilde{w}_{k-1}^T \end{bmatrix}^T$}$, $\tilde{v}_{0:k} \triangleq \text{\small$\begin{bmatrix} \tilde{v}_0^T \text{ } \cdots \text{ } \tilde{v}_k^T \end{bmatrix}^T$}$, $H_A\triangleq H_D\text{BlkDiag}(\bar{A}_0,\cdots,\bar{A}_k)$, $H_B\triangleq H_D\text{BlkDiag}(\tilde{B}_0,\cdots,\tilde{B}_k)$,  and $H_C\triangleq\text{BlkDiag}(\bar{C}_0,\cdots,\bar{C}_k)$ with
\begin{equation*} \small
\medmuskip=2.58mu
\thinmuskip=2.58mu
\thickmuskip=2.58mu
H_D \triangleq
\begin{bmatrix}
0 & \cdots & 0 & 0 \\
\tilde{C}\tilde{A}^0 & \cdots & 0 & 0 \\
\vdots & \ddots & \vdots & \vdots \\
\tilde{C}\tilde{A}^{k-1} & \cdots & \tilde{C}\tilde{A}^0 & 0
\end{bmatrix},
H_W \triangleq
\begin{bmatrix}
\tilde{C}\tilde{A}^0 & \cdots & 0 \\
\vdots & \ddots & \vdots \\
\tilde{C}\tilde{A}^k & \cdots & \tilde{C}\tilde{A}^0
\end{bmatrix}.
\end{equation*}

To quantify the amount of information the defender receives about the attacked states through the biased auxiliary sensor measurements, we use the Fisher information matrix $\mathcal{I}$. The Fisher information matrix is a metric that quantifies the amount of information a set of measurements contains about a set of unknown parameters. As seen in \cite{acc2017griffioen}, the Fisher information matrix is shown to be
\begin{equation} \small
\label{LinearACFIM}
\mathcal{I} = (H_A + H_C)^T (H_W \Sigma_{\tilde{Q}} H_W^T + \Sigma_{\tilde{R}})^{-1} (H_A + H_C),
\end{equation}
where $\tilde{x}_0\sim\mathcal{N}(0,\tilde{P}_{0|-1})$, $\tilde{P}_{k+1|k}$ is the a priori error covariance matrix for the auxiliary system, $\Sigma_{\tilde{Q}}\triangleq\text{BlkDiag}(\tilde{P}_{0|-1},\tilde{Q},\cdots,\tilde{Q})$, and $\Sigma_{\tilde{R}}\triangleq\text{BlkDiag}(\tilde{R},\cdots,\tilde{R})$.

We note that $\mathcal{I}$ is positive semidefinite, allowing us to maximize the expected value of the Fisher information matrix for all possible additive integrity attacks. This is carried out by maximizing a nonnegative constant $\gamma$ such that the expected value of the Fisher information matrix is greater than a positive semidefinite lower bound $\Theta(\gamma)$ that is a function of $\gamma$. Since there are real-world constraints on the variance magnitude of the state coupling and auxiliary sensors, we constrain the covariances $\Sigma_{\bar{A}}$ and $\Sigma_{\bar{C}}$ with positive semidefinite upper bounds $\Theta_A$ and $\Theta_C$, respectively. This maximization problem is presented below
\begin{equation} \small
\medmuskip=1.8mu
\thinmuskip=1.8mu
\thickmuskip=1.8mu
\label{InitialBarOptimization}
\argmax_{\gamma,\Sigma_{\bar{A}},\Sigma_{\bar{C}}} \gamma ~~ \text{s.t. } \Sigma_{\bar{A}} \preceq \Theta_A, ~ \Sigma_{\bar{C}} \preceq \Theta_C, ~ \mathbb{E}_{\bar{A}_{0;k-1},\bar{C}_{0:k}}\left[\mathcal{I}\right] \succeq \Theta(\gamma).
\end{equation}

As shown in \cite{acc2017griffioen}, we can construct the positive semidefinite lower bound $\Theta(\gamma)$ to match the block structure of $\mathcal{I}$. The off-diagonal blocks of this structure are not functions of $\Sigma_{\bar{A}}$ or $\Sigma_{\bar{C}}$, allowing the constraint in (\ref{InitialBarOptimization}) to be simplified to the following series of constraints
\begin{equation} \small
\label{ConstraintBarV1}
\Omega_A^i + \Omega_C^i + \Omega_{AC}^i + \Omega_{AC}^{i^T} \succeq \gamma \Theta_i, \quad i = 0,\cdots,k,
\end{equation}
where $\Theta_i\succeq0$, $\Omega_A^i=\text{Tr}(J_{ii})\Sigma_{\bar{A}}+\text{Sum}(J_{ii})\mu_{\bar{A}}\mu_{\bar{A}}^T$, $\Omega_C^i=\text{Tr}(S_{ii})\Sigma_{\bar{C}}+\text{Sum}(S_{ii})\mu_{\bar{C}}\mu_{\bar{C}}^T$, and $\Omega_{AC}^i\triangleq\mathbb{E}_{\bar{A}_i}\left[\bar{A}_i^T\right]F_{ii}\mathbb{E}_{\bar{C}_i}\left[\bar{C}_i\right]$. $J_{ii}\in\mathbb{R}^{\tilde{n}\times\tilde{n}}$, $S_{ii}\in\mathbb{R}^{\tilde{m}\times\tilde{m}}$, and $F_{ii}\in\mathbb{R}^{\tilde{n}\times\tilde{m}}$ represent the $(i,i)$th blocks of $H_D^T(H_W\Sigma_{\tilde{Q}}H_W^T+\Sigma_{\tilde{R}})^{-1}H_D$, $(H_W\Sigma_{\tilde{Q}}H_W^T+\Sigma_{\tilde{R}})^{-1}$, and $H_D^T(H_W\Sigma_{\tilde{Q}}H_W^T+\Sigma_{\tilde{R}})^{-1}$, respectively.

Utilizing (\ref{ConstraintBarV1}) and the results above allows the third constraint in (\ref{InitialBarOptimization}) to be written as a series of positive semidefinite constraints. We choose $k=T-1$ so that the Fisher information matrix is maximized over the time window $T$ of the chi-squared detector. This allows the optimization problem in (\ref{InitialBarOptimization}) to be written as
\begin{equation} \small
\label{ACOptimization}
\begin{split}
& \argmax_{\gamma,\Sigma_{\bar{A}},\Sigma_{\bar{C}}} \gamma \quad \text{s.t. } \Sigma_{\bar{A}} \preceq \Theta_A, \quad \Sigma_{\bar{C}} \preceq \Theta_C,  \\
& \quad\quad \text{Tr}(J_{ii})\Sigma_{\bar{A}} + \text{Tr}(S_{ii})\Sigma_{\bar{C}} + \text{Sum}(J_{ii})\mu_{\bar{A}}\mu_{\bar{A}}^T \\
& \quad\quad + \text{Sum}(S_{ii})\mu_{\bar{C}}\mu_{\bar{C}}^T + \begin{bmatrix}\mu_{\bar{A}}\text{ }\cdots\text{ }\mu_{\bar{A}}\end{bmatrix}F_{ii}\begin{bmatrix}\mu_{\bar{C}}\text{ }\cdots\text{ }\mu_{\bar{C}}\end{bmatrix}^T \\
& \quad\quad + \begin{bmatrix}\mu_{\bar{C}}\text{ }\cdots\text{ }\mu_{\bar{C}}\end{bmatrix}F_{ii}^T\begin{bmatrix}\mu_{\bar{A}}\text{ }\cdots\text{ }\mu_{\bar{A}}\end{bmatrix}^T \succeq \gamma \Theta_i,
\end{split}
\end{equation}
which is a semidefinite program with $i=0,\cdots,T-1$ and unique solutions for $\gamma$, $\Sigma_{\bar{A}}$, and $\Sigma_{\bar{C}}$. As a result, this optimization problem provides a method for designing $\Sigma_{\bar{A}}$ and $\Sigma_{\bar{C}}$ so that the expected value of the Fisher information matrix is maximized for the set of all possible additive integrity attacks described by (\ref{AttackStateSensorDynamics}).
\begin{remark}
Unlike the hybrid moving target defense, the extended moving target defense as presented is not designed directly to perform attack identification or resilient state estimation. The nominal sensors only measure the original LTI dynamics and as such offer no inherent advantages over standard techniques for identification and estimation. Nevertheless, if the sensors for the extended subsystem are secure, an avenue exists for resilient state estimation. Specifically, consider a system with (only) sensor attacks on the nominal sensors. If the system consisting of the nominal dynamics with the extended sensors is observable, a resilient state estimate can be obtained that can be used to determine which of the nominal sensors are reporting false or misleading information, resulting in attack identification.
\end{remark}
\section{Nonlinear Moving Target Defense}
While the parameters that generate the time-varying dynamics of the extended moving target are designed to maximize detection and estimation performance, it is still possible for an intelligent adversary to perform some system identification. This is due to the fact that the auxiliary sensor measurements contain some information about the time-varying matrices $\bar{A}_k$, $\tilde{B}_k$, and $\bar{C}_k$. As a result, we seek to leverage the nonlinearity $G_kh(x_k)$ in the nonlinear moving target to minimize the amount of information an adversary may receive about the time-varying matrices $\bar{A}_k$, $\tilde{B}_k$, and $\bar{C}_k$. The auxiliary sensors in this system measure a nonlinear function of the state, where the nonlinear function $h(x_k)$ is an element-wise mapping from $\mathbb{R}^n\to\mathbb{R}^n$ and $G_k\in\mathbb{R}^{\tilde{m}\times n}$ is generated from the distribution $G_k(\text{column } i)\sim\mathcal{N}(\mu_G,\Sigma_{G})\text{ }\forall i$ with independence between columns over time.

\subsection{Limiting System Identification}
We again consider a general set of integrity attacks as modeled in (\ref{AttackStateSensorDynamics}) for the nonlinear moving target defense. From the perspective of the attacker, this can be written as
\begin{align}
\label{NonlinearStateDynamics}
\text{\small$\bar{x}_{k+1}^A$} &\text{\small$=\mathcal{A}_k\bar{x}_k^A + \mathcal{B}_k\underbrace{(u_k+u_k^a)}_{u_k^A} + \bar{w}_k,$} \\
\label{NonlinearSensorDynamics}
\text{\small$\underbrace{\begin{bmatrix}
\tilde{y}_k^A \\
y_k^A
\end{bmatrix}}_{\bar{y}_k^A}$} &\text{\small$=
\underbrace{\begin{bmatrix}
\tilde{C} & \bar{C}_k \\
0 & C
\end{bmatrix}}_{\mathcal{C}_k}
\underbrace{\begin{bmatrix}
\tilde{x}_k^A \\
x_k^A
\end{bmatrix}}_{\bar{x}_k^A} +
\begin{bmatrix}
G_kh(x_k^A)\\
0
\end{bmatrix} +
\underbrace{\begin{bmatrix}
\tilde{v}_k \\
v_k
\end{bmatrix}}_{\bar{v}_k},$}
\end{align}
where $\bar{y}_k^A$ represents the sensor measurements that the attacker intercepts. Given these dynamics, the auxiliary intercepted sensor measurements are given by
\begin{equation} \small
\begin{split}
\tilde{y}_k^A & = \tilde{C}\tilde{A}^k\tilde{x}_0 + \tilde{C}\sum_{j=0}^{k-1}\tilde{A}^{k-1-j}(\bar{A}_jx_j^A+\tilde{B}_ju_j^A+\tilde{w}_j) \\
& \quad + \bar{C}_kx_k^A + G_kh(x_k^A) + \tilde{v}_k.
\end{split}
\end{equation}
Considering the amount of information all the auxiliary intercepted sensor measurements $\tilde{y}_{0:k}^A$ contain about all the time-varying matrices $\bar{A}_{0:k}$, $\tilde{B}_{0:k}$, and $\bar{C}_{0:k}$, we can represent all the auxiliary intercepted sensor measurements $\tilde{y}_{0:k}^A$ as
\begin{equation} \small
\begin{split}
\tilde{y}_{0:k}^A & = H_X\text{vec}(\bar{A}_{0:k}^T) + H_U\text{vec}(\tilde{B}_{0:k}^T) + H_E\text{vec}(\bar{C}_{0:k}^T) \\
& \quad + H_F\text{vec}(G_{0:k}) + H_W\tilde{w}_{0:k-1} + \tilde{v}_{0:k},
\end{split}
\end{equation}
where $\tilde{y}_{0:k}^A \triangleq \text{\small$\begin{bmatrix} \tilde{y}_0^{A^T} \text{ } \cdots \text{ } \tilde{y}_k^{A^T} \end{bmatrix}^T$}$, $\text{\small$\text{vec}(\bar{A}_{0:k}^T)\triangleq\begin{bmatrix}\text{vec}(\bar{A}_0^T)^T\text{ }\cdots\text{ }\text{vec}(\bar{A}_k^T)^T\end{bmatrix}^T$}$, $\text{\small$\text{vec}(\tilde{B}_{0:k}^T)\triangleq\begin{bmatrix}\text{vec}(\tilde{B}_0^T)^T\text{ }\cdots\text{ }\text{vec}(\tilde{B}_k^T)^T\end{bmatrix}^T$}$, $\text{\small$\text{vec}(\bar{C}_{0:k}^T)\triangleq\begin{bmatrix}\text{vec}(\bar{C}_0^T)^T\text{ }\cdots\text{ }\text{vec}(\bar{C}_k^T)^T\end{bmatrix}^T$}$, and $\text{\small$\text{vec}(G_{0:k})\triangleq\begin{bmatrix}\text{vec}(G_0)^T\text{ }\cdots\text{ }\text{vec}(G_k)^T\end{bmatrix}^T$}$ and the matrices are given by $H_X\triangleq H_D\text{BlkDiag}(I_{\tilde{n}}\otimes x_0^{A^T},\cdots,I_{\tilde{n}}\otimes x_k^{A^T})$, $H_U\triangleq H_D\text{BlkDiag}(I_{\tilde{n}}\otimes u_0^{A^T},\cdots,I_{\tilde{n}}\otimes u_k^{A^T})$, $H_E\triangleq\text{BlkDiag}\left(I_{\tilde{m}}\otimes x_0^{A^T},\cdots,I_{\tilde{m}}\otimes x_k^{A^T}\right)$, and $H_F\triangleq\text{BlkDiag}\left(h(x_0^A)^T\otimes I_{\tilde{m}},\cdots,h(x_k^A)^T\otimes I_{\tilde{m}}\right)$.

We consider a strong adversary who has full knowledge of the attacked states $x_{0:k}^A$. With this adversary, the distribution of all the auxiliary intercepted sensor measurements given the time-varying parameters $\theta_{0:k}\triangleq\text{\small$\begin{bmatrix}\text{vec}(\bar{A}_{0:k}^T)^T&\text{vec}(\tilde{B}_{0:k}^T)^T&\text{vec}(\bar{C}_{0:k}^T)^T\end{bmatrix}^T$}$ follows a normal distribution $\tilde{y}_{0:k}^A|\theta_{0:k}\sim\mathcal{N}(\mu_{\tilde{y}_{0:k}^A|\theta_{0:k}},\Sigma_{\tilde{y}_{0:k}^A|\theta_{0:k}})$ with mean and covariance given by
\begin{equation} \small
\mu_{\tilde{y}_{0:k}^A|\theta_{0:k}} = H_X\text{vec}(\bar{A}_{0:k}^T) + H_U\text{vec}(\tilde{B}_{0:k}^T) + H_E\text{vec}(\bar{C}_{0:k}^T) + H_F\mu_{\vec{g}},
\end{equation}
\begin{equation} \small
\Sigma_{\tilde{y}_{0:k}^A|\theta_{0:k}} = H_F \Sigma_{\vec{g}} H_F^T + H_W \Sigma_{\tilde{Q}} H_W^T + \Sigma_{\tilde{R}},
\end{equation}
where $\mu_{\vec{g}}\triangleq\begin{bmatrix}\mu_G^T\text{ }\cdots\text{ }\mu_G^T\end{bmatrix}^T$ and $\Sigma_{\vec{g}}\triangleq\text{BlkDiag}(\Sigma_G,\cdots,\Sigma_G)$.

To quantify the amount of information the strong adversary receives about the time-varying parameters $\theta_{0:k}$ from the auxiliary intercepted sensor measurements $\tilde{y}_{0:k}^A$, we use the Bayesian Fisher information matrix $\mathcal{I}_{NL}$ which accounts for prior information the strong adversary has access to about the time-varying parameters. Since the joint distribution of $\tilde{y}_{0:k}^A$ and $\theta_{0:k}$ is Gaussian, each element $\mathcal{I}_{NL}(i,j)$ of the Fisher information matrix takes the following form \cite{FIM}
\begin{equation} \small
\medmuskip=1.57mu
\thinmuskip=1.57mu
\thickmuskip=1.57mu
\begin{split}
\mathcal{I}_{NL}(i,j) =& \mathbb{E}_{\theta_{0:k}}\Bigg[\frac{\partial \mu_{\tilde{y}_{0:k}^A|\theta_{0:k}}^T}{\partial \theta_{0:k}(i)}\Sigma_{\tilde{y}_{0:k}^A|\theta_{0:k}}^{-1}\frac{\partial \mu_{\tilde{y}_{0:k}^A|\theta_{0:k}}}{\partial \theta_{0:k}(j)} + \frac{1}{2}\Sigma_{\theta_{0:k}}^{-1}(i,j) \\
& + \frac{1}{2}\text{Tr}\left(\Sigma_{\tilde{y}_{0:k}^A|\theta_{0:k}}^{-1}\frac{\partial \Sigma_{\tilde{y}_{0:k}^A|\theta_{0:k}}}{\partial \theta_{0:k}(i)}\Sigma_{\tilde{y}_{0:k}^A|\theta_{0:k}}^{-1}\frac{\partial \Sigma_{\tilde{y}_{0:k}^A|\theta_{0:k}}}{\partial \theta_{0:k}(j)}\right)\Bigg],
\end{split}
\end{equation}
where $\Sigma_{\theta_{0:k}}\triangleq\text{BlkDiag}(\Sigma_{\bar{A}},\cdots,\Sigma_{\bar{A}},\Sigma_{\tilde{B}},\cdots,\Sigma_{\tilde{B}},$ $\Sigma_{\bar{C}},\cdots,\Sigma_{\bar{C}})$, and the partial derivatives of the mean $\mu_{\tilde{y}_{0:k}^A|\theta_{0:k}}$ and the covariance $\Sigma_{\tilde{y}_{0:k}^A|\theta_{0:k}}$ equal $0$ and $H(\text{column }i)$, respectively, with $H\triangleq\begin{bmatrix}H_X&H_U&H_E\end{bmatrix}$. Applying these results to each element of the Fisher information matrix implies that $\mathcal{I}_{NL}$ can be written as
\begin{equation} \small
\medmuskip=1.95mu
\thinmuskip=1.95mu
\thickmuskip=1.95mu
\label{NonlinearFIM}
\begin{split}
& \mathcal{I}_{NL} = H^T (H_F \Sigma_{\vec{g}} H_F^T + H_W \Sigma_{\tilde{Q}} H_W^T + \Sigma_{\tilde{R}})^{-1} H + \frac{1}{2}\Sigma_{\theta_{0:k}}^{-1} \\
& = H^T \Bigg(\text{BlkDiag}\left(\sum_{i=1}^nh(x_0^A(i))^2\Sigma_G,\cdots,\sum_{i=1}^nh(x_k^A(i))^2\Sigma_G\right) \\
& \quad + H_W \Sigma_{\tilde{Q}} H_W^T + \Sigma_{\tilde{R}}\Bigg)^{-1} H + \frac{1}{2}\Sigma_{\theta_{0:k}}^{-1}.
\end{split}
\end{equation}
To understand how the nonlinear term $G_kh(x_k^A)$ influences the Fisher information matrix, we consider the dynamics given in (\ref{NonlinearStateDynamics}) and (\ref{NonlinearSensorDynamics}) without the nonlinear term $G_kh(x_k^A)$ and see that the Fisher information matrix $\mathcal{I}_L$ can be written as
\begin{equation} \small
\label{LinearFIM}
\mathcal{I}_L = H^T (H_W \Sigma_{\tilde{Q}} H_W^T + \Sigma_{\tilde{R}})^{-1} H + \frac{1}{2}\Sigma_{\theta_{0:k}}^{-1}.
\end{equation}

Given these representations of $\mathcal{I}_{NL}$ and $\mathcal{I}_L$ in (\ref{NonlinearFIM}) and (\ref{LinearFIM}), we can use the Woodbury identity to show that the difference between the Fisher information matrices associated with the systems containing and not containing the nonlinearity is positive definite as seen below
\begin{equation} \small
\label{FIMdifference}
\begin{split}
\mathcal{I}_L - &\mathcal{I}_{NL} = H^T \Sigma_N^{-1}H_F(\Sigma_{\vec{g}}^{-1}+H_F^T\Sigma_N^{-1}H_F)^{-1} H_F^T\Sigma_N^{-1} H \\
& = H^T \Sigma_N^{-1}\Bigg(\text{BlkDiag}\bigg(\frac{1}{\sum_{i=1}^nh(x_0^A(i))^2}\Sigma_G^{-1},\cdots \\
& \quad \cdots, \frac{1}{\sum_{i=1}^nh(x_k^A(i))^2}\Sigma_G^{-1}\bigg) + \Sigma_N^{-1}\Bigg)^{-1} \Sigma_N^{-1} H \succeq 0,
\end{split}
\end{equation}
where $\Sigma_N\triangleq H_W\Sigma_{\tilde{Q}}H_W^T+\Sigma_{\tilde{R}}$. Because this difference is positive definite, adding the nonlinear term $G_kh(x_k^A)$ to the system dynamics decreases the amount of information the strong adversary receives about the time-varying parameters $\theta_{0:k}$ from the auxiliary intercepted sensor measurements $\tilde{y}_{0:k}^A$.

To minimize this amount of information, we would like to design $\Sigma_G$ and the function $h$ to maximize $\mathcal{I}_L - \mathcal{I}_{NL}$, the difference between the Fisher information matrices associated with the systems containing and not containing the nonlinearity. From (\ref{NonlinearFIM}) and (\ref{FIMdifference}), we see that as $\|\Sigma_G\|\to\infty$ or as $h(x_k^A)\to\infty$, $\mathcal{I}_{NL}\to\frac{1}{2}\Sigma_{\theta_{0:k}}^{-1}$ and $\mathcal{I}_L-\mathcal{I}_{NL}\to H^T(H_W\Sigma_{\tilde{Q}}H_W^T+\Sigma_{\tilde{R}})^{-1}H$. As the covariance $\Sigma_G$ or nonlinear function $h(x_k^A)$ approaches infinity, the information the adversary receives about the time-varying parameters $\theta_{0:k}$ is reduced to his or her a priori information about $\theta_{0:k}$. As a result, an increase in the covariance $\Sigma_G$ or an increase in the magnitude of the nonlinear function $h(x_k^A)$ results in the adversary receiving less information about the time-varying parameters $\theta_{0:k}$. In addition, as $\Sigma_G\to0$ or as $h(x_k^A)\to0$, $\mathcal{I}_{NL}\to\mathcal{I}_L$ and $\mathcal{I}_L-\mathcal{I}_{NL}\to0$. Consequently, a decrease in the covariance $\Sigma_G$ or a decrease in the magnitude of $h(x_k^A)$ results in the adversary receiving more information about the time-varying parameters $\theta_{0:k}$.

This general analysis provides intuition about the effects of the magnitude of the nonlinearity on the information received by the adversary. Because the function $h(x_k)$ determines the magnitude of the nonlinearity while the coefficient matrix $G_k$ determines the direction of the nonlinearity, in Section VII we design the function $h(x_k)$ to limit the adversary's information while in the next subsection we design the coefficient matrix $G_k$ to maximize the defender's estimation performance.

\subsection{Nonlinearity Design}
We now consider the same additive integrity attacks on the nonlinear moving target as described in (\ref{NonlinearStateDynamics}) and (\ref{NonlinearSensorDynamics}) from the perspective of the defender. We want to provide joint guidelines for designing the parameters of the distributions that generate the time-varying matrices $\bar{A}_k$, $\bar{C}_k$, and $G_k$. Here we do not design the parameters that generate $\tilde{B}_k$ because the extended moving target design of $\Sigma_{\tilde{B}_k}$ should be sufficient to maximize detection performance as long as the magnitude of the nonlinear function $h$ is approximately zero when the state lies within a normal region of operation. In providing a joint design for the time-varying matrices $\bar{A}_k$, $\bar{C}_k$, and $G_k$, we seek to design $\Sigma_{\bar{A}}$, $\Sigma_{\bar{C}}$, and $\Sigma_G$ to maximize the amount of information the defender receives about the attacked states to improve the defender's estimation performance, in turn improving detection performance. The state dynamics are the same as those given in (\ref{NonlinearStateDynamics}), while the sensor measurements are given by
\begin{equation} \small
\label{NonlinearSensorDynamics2}
\underbrace{\begin{bmatrix}
\tilde{y}_k^a \\
y_k^a
\end{bmatrix}}_{\bar{y}_k^a} =
\underbrace{\begin{bmatrix}
\tilde{C} & \bar{C}_k \\
0 & C
\end{bmatrix}}_{\mathcal{C}_k}
\underbrace{\begin{bmatrix}
\tilde{x}_k^A \\
x_k^A
\end{bmatrix}}_{\bar{x}_k^A} +
\begin{bmatrix}
G_kh(x_k^A)\\
0
\end{bmatrix} +
\underbrace{\begin{bmatrix}
\tilde{v}_k \\
v_k
\end{bmatrix}}_{\bar{v}_k} +
\underbrace{\begin{bmatrix}
\tilde{d}_k^a \\
d_k^a
\end{bmatrix}}_{\bar{d}_k^a},
\end{equation}
where $\bar{y}_k^a$ represents the biased sensor measurements that the defender receives. Given these dynamics, the biased auxiliary sensor measurements are given by
\begin{equation} \small
\begin{split}
\tilde{y}_k^a & = \tilde{C}\tilde{A}^k\tilde{x}_0 + \tilde{C}\sum_{j=0}^{k-1}\tilde{A}^{k-1-j}\left(\bar{A}_jx_j^A+\tilde{B}_ju_j^A+\tilde{w}_j\right) \\
& \quad + \bar{C}_kx_k^A + G_kh(x_k^A) + \tilde{v}_k + \tilde{d}_k^a.
\end{split}
\end{equation}
Considering the amount of information all the biased auxiliary sensor measurements $\tilde{y}_{0:k}^a$ carry about all the attacked states $x_{0:k}^A$, we can represent all the biased auxiliary sensor measurements $\tilde{y}_{0:k}^a$ as
\begin{equation} \small
\begin{split}
\tilde{y}_{0:k}^a & = (H_A+H_C)x_{0:k}^A + H_Gh(x_{0:k}^A) + H_Bu_{0:k}^A \\
& \quad + \tilde{d}_{0:k}^a + H_W\tilde{w}_{0:k-1} + \tilde{v}_{0:k},
\end{split}
\end{equation}
where $H_G\triangleq\text{BlkDiag}(G_0,\cdots,G_k)$ and $h(x_{0:k}^A)\triangleq\text{\small$\begin{bmatrix}h(x_0^A)^T\text{ }\cdots\text{ }h(x_k^A)^T\end{bmatrix}^T$}$. With this representation, we see that the distribution of all the biased auxiliary sensor measurements given all the attacked states follows a normal distribution $\tilde{y}_{0:k}^a|x_{0:k}^A \sim \mathcal{N}(\mu_{\tilde{y}_{0:k}^a|x_{0:k}^A},\Sigma_{\tilde{y}_{0:k}^a|x_{0:k}^A})$ with mean and covariance given by
\begin{equation} \small
\mu_{\tilde{y}_{0:k}^a|x_{0:k}^A} = (H_A+H_C)x_{0:k}^A + H_Gh(x_{0:k}^A) + H_Bu_{0:k}^A + \tilde{d}_{0:k}^a,
\end{equation}
\begin{equation} \small
\Sigma_{\tilde{y}_{0:k}^a|x_{0:k}^A} = H_W \Sigma_{\tilde{Q}} H_W^T + \Sigma_{\tilde{R}}.
\end{equation}

To quantify the amount of information the defender receives about the attacked states through the biased auxiliary sensor measurements, we use the Fisher information matrix $\bar{\mathcal{I}}$. Since $\tilde{y}_{0:k}^a|x_{0:k}^A$ follows a multivariate Gaussian distribution, each element of the Fisher information matrix $\bar{\mathcal{I}}(i,j)$ takes the following form \cite{FIM}
\begin{equation} \small
\begin{split}
& \bar{\mathcal{I}}(i,j) = \frac{\partial \mu_{\tilde{y}_{0:k}^a|x_{0:k}^A}^T}{\partial x_{0:k}^A(i)}\Sigma_{\tilde{y}_{0:k}^a|x_{0:k}^A}^{-1}\frac{\partial \mu_{\tilde{y}_{0:k}^a|x_{0:k}^A}}{\partial x_{0:k}^A(j)} + \\
& \quad + \frac{1}{2}\text{Tr}\left(\Sigma_{\tilde{y}_{0:k}^a|x_{0:k}^A}^{-1}\frac{\partial \Sigma_{\tilde{y}_{0:k}^a|x_{0:k}^A}}{\partial x_{0:k}^A(i)}\Sigma_{\tilde{y}_{0:k}^a|x_{0:k}^A}^{-1}\frac{\partial \Sigma_{\tilde{y}_{0:k}^a|x_{0:k}^A}}{\partial x_{0:k}^A(j)}\right),
\end{split}
\end{equation}
where the partial derivative of the covariance $\Sigma_{\tilde{y}_{0:k}^a|x_{0:k}^A}$ equals $0$ and the partial derivative of the mean is given by
\begin{equation} \small
\frac{\partial \mu_{\tilde{y}_{0:k}^a|x_{0:k}^A}}{\partial x_{0:k}^A(i)} = H_A(\text{clmn }i) + H_C(\text{clmn }i) + H_G\frac{\partial h(x_{0:k}^A)}{\partial x_{0:k}^A(i)}.
\end{equation}
Applying these results to each element of the Fisher information matrix implies that $\bar{\mathcal{I}}$ can be written as
\begin{equation} \small
\begin{split}
\bar{\mathcal{I}} = & \left(H_A + H_C + H_G\frac{\partial h(x_{0:k}^A)}{\partial x_{0:k}^A}\right)^T (H_W \Sigma_{\tilde{Q}} H_W^T + \\
& + \Sigma_{\tilde{R}})^{-1} \left(H_A + H_C + H_G\frac{\partial h(x_{0:k}^A)}{\partial x_{0:k}^A}\right).
\end{split}
\end{equation}

To maximize the amount of information the defender receives about the attacked states, we want to maximize the expected value of the Fisher information matrix which takes the following form
\begin{equation} \small
\medmuskip=-0.05mu
\thinmuskip=-0.05mu
\thickmuskip=-0.05mu
\label{ExpectedNonlinearFIM}
\begin{split}
& \mathbb{E}_{\bar{A}_{0:k-1},\bar{C}_{0:k},G_{0:k}}\left[\bar{\mathcal{I}}\right] = \mathbb{E}_{\bar{A}_{0:k-1},\bar{C}_{0:k}}\left[\Omega_A + \Omega_C + \Omega_{AC} + \Omega_{AC}^T\right] \\
& \quad\quad\quad + \mathbb{E}_{\bar{A}_{0:k-1},\bar{C}_{0:k},G_{0:k}}\left[\Omega_G + \Omega_{AG} + \Omega_{AG}^T + \Omega_{CG} + \Omega_{CG}^T\right],
\end{split}
\end{equation}
where $\Omega_A$, $\Omega_C$, and $\Omega_{AC}$ are defined as given in (\ref{ConstraintBarV1}) and $\Omega_G \triangleq \frac{\partial h(x_{0:k}^A)^T}{\partial x_{0:k}^A}H_G^T\Sigma_N^{-1}H_G\frac{\partial h(x_{0:k}^A)}{\partial x_{0:k}^A}$, $\Omega_{AG} \triangleq H_A^T\Sigma_N^{-1}H_G\frac{\partial h(x_{0:k}^A)}{\partial x_{0:k}^A}$, and $\Omega_{CG} \triangleq H_C^T\Sigma_N^{-1}H_G\frac{\partial h(x_{0:k}^A)}{\partial x_{0:k}^A}$. We note that the first term in (\ref{ExpectedNonlinearFIM}) contains $\Sigma_{\bar{A}}$ and $\Sigma_{\bar{C}}$ while the second term only contains $\Sigma_G$. Furthermore, the first term in (\ref{ExpectedNonlinearFIM}) is simply the expected value of the Fisher information matrix given in (\ref{LinearACFIM}). Consequently, we can maximize the amount of information the defender receives about the attacked states by jointly designing $\Sigma_{\bar{A}}$ and $\Sigma_{\bar{C}}$ as given in (\ref{ACOptimization}) while designing $\Sigma_G$ to maximize the second term in (\ref{ExpectedNonlinearFIM}).

Because $\mathbb{E}_{\bar{A}_{0:k-1},\bar{C}_{0:k},G_{0:k}}\left[\Omega_{AG} + \Omega_{AG}^T + \Omega_{CG} + \Omega_{CG}^T\right]$ does not contain $\Sigma_G$, we only consider $\mathbb{E}_{G_{0:k}}\left[\Omega_G\right]$ when maximizing the second term in (\ref{ExpectedNonlinearFIM}). Furthermore, only the diagonal entries of $\mathbb{E}_{G_{0:k}}\left[\Omega_G\right]$ are functions of $\Sigma_G$, so we consider $\text{Tr}\left(\mathbb{E}_{G_{0:k}}\left[\Omega_G\right]\right)$ as the metric to be maximized where $\text{Tr}\left(\mathbb{E}_{G_{0:k}}\left[\Omega_G\right]\right)=x_{NL}^{A^T}\mathbb{E}_{G_{0:k}}\left[\mathcal{F}_G\right]x_{NL}^A$ with $x_{NL}^A\triangleq\text{\small$\begin{bmatrix}\frac{\partial h(x_{0:k}^A(1))}{\partial x_{0:k}^A(1)}\text{ }\cdots\text{ }\frac{\partial h(x_{0:k}^A((k+1)n))}{\partial x_{0:k}^A((k+1)n)}\end{bmatrix}^T$}$ and $\mathcal{F}_G\triangleq\text{Diag}(H_G(\text{clmn }1)^T\Sigma_N^{-1}H_G(\text{clmn }1),\cdots,H_G(\text{clmn }(k+1)n)^T\Sigma_N^{-1}H_G(\text{clmn }(k+1)n))$. Because $x_{NL}^A$ is unknown to the defender, it becomes difficult to maximize $\text{Tr}\left(\mathbb{E}_{G_{0:k}}\left[\Omega_G\right]\right)$. However, we note that $\mathbb{E}_{G_{0:k}}\left[\mathcal{F}_G\right]$ is positive semidefinite, allowing us to maximize $\text{Tr}\left(\mathbb{E}_{G_{0:k}}\left[\Omega_G\right]\right)$ for all possible $x_{NL}^A$ by maximizing a nonnegative constant $\beta$ such that $\mathbb{E}_{G_{0:k}}\left[\mathcal{F}_G\right]$ is greater than a positive semidefinite lower bound $\beta I$. This maximization problem is presented below where $\mathcal{M}$ is a positive semidefinite upper bound that represents real-world constraints on the variance magnitude of the nonlinear auxiliary sensors
\begin{equation} \small
\label{InitialNonlinearOptimization}
\argmax_{\beta,\Sigma_G} \beta \quad \text{s.t. } \Sigma_G \preceq \mathcal{M}, \text{ } \mathbb{E}_{G_{0:k}}\left[\mathcal{F}_G\right] \succeq \beta I.
\end{equation}

Since $\mathcal{F}_G$ is a diagonal matrix, the second constraint in (\ref{InitialNonlinearOptimization}) can be simplified to the following series of constraints
\begin{equation} \small
\mathbb{E}_{G_{0:k}}\left[\mathcal{F}_G(i,i)\right] \geq \beta, \quad i=1,\cdots,(k+1)n.
\end{equation}
Noting that $\mathbb{E}_{G_{0:k}}\left[\mathcal{F}_G(i,i)\right]=\text{Tr}\left((\Sigma_G+\mu_G\mu_G^T)S_{jj}\right)$ where $j\triangleq\lfloor(i-1)/n\rfloor$ and choosing $k=T-1$ so that $\text{Tr}\left(\mathbb{E}_{G_{0:k}}\left[\Omega_G\right]\right)$ is maximized over the time window $T$ of the chi-squared detector, the optimization problem in (\ref{InitialNonlinearOptimization}) can be written as
\begin{equation} \small
\label{GkOptimization}
\argmax_{\beta,\Sigma_G} \beta \quad \text{s.t. } \Sigma_G \preceq \mathcal{M}, \text{ Tr}\left((\Sigma_G+\mu_G\mu_G^T)S_{ii}\right) \geq \beta,
\end{equation}
with $i = 0,\cdots,T-1$ and unique solutions for $\beta$ and $\Sigma_G$. Consequently, this optimization problem provides a method for designing $\Sigma_G$ to maximize the expected value of the Fisher information matrix for all possible additive integrity attacks described by (\ref{NonlinearStateDynamics}) and (\ref{NonlinearSensorDynamics2}).

\section{Bounds on Attacker's Performance}
We now turn our attention to calculating lower bounds on the detection statistic associated with optimal attacks on each of the moving target systems. These bounds characterize the worst case detection performance while under attack. We first investigate lower bounds on the attacker's state estimation and use these bounds to understand how well an adversary can fool the bad data detector.

\subsection{Attack Strategy}
We consider an attack strategy where the adversary aims to track the system operator's state estimate $\hat{\bar{x}}_{k|k-1}$. By tracking the system operator's state estimate, the adversary attempts to generate stealthy outputs. We assume the adversary has full knowledge of the nominal static system model, is able to read and modify all the control inputs and all the sensor outputs, and knows the probability density function (pdf) of the random matrices and the noise. Without loss of generality, we assume that the attack begins at $k=0$. For the hybrid moving target, the attacker's observations and strategy are formulated as
\begin{equation} \small
\label{HybridAttackStrategyStates}
\begin{split}
\begin{bmatrix}
x_{k+1}^A \\
\hat{x}_{k+1|k}
\end{bmatrix} & =
\begin{bmatrix}
A_k & 0 \\
0 & A_k(I-K_kC_k)
\end{bmatrix}
\begin{bmatrix}
x_k^A \\
\hat{x}_{k|k-1}
\end{bmatrix} \\ &\quad +
\begin{bmatrix}
B_k & B_k & 0 \\
B_k & 0 & A_kK_k
\end{bmatrix}
\begin{bmatrix}
u_k \\
u_k^a \\
y_k^a
\end{bmatrix} +
\begin{bmatrix}
w_k \\
0
\end{bmatrix},
\end{split}
\end{equation}
\begin{equation} \small
\label{HybridAttackStrategySensors}
y_k^A =
\begin{bmatrix}
C_k & 0
\end{bmatrix}
\begin{bmatrix}
x_k^A \\
\hat{x}_{k|k-1}
\end{bmatrix} + v_k,
~ d_k^a = \mathbb{E}\left[C_k\hat{x}_{k|k-1}\middle|\mathcal{I}_k^A\right] - y_k^A,
\end{equation}
where $x_k^A$ represents the attacked states for the nominal system, $y_k^A$ denotes the sensor measurements for the nominal system that the attacker intercepts, $y_k^a$ represents the biased sensor measurements for the nominal system received by the system operator, and $u_k^a=0$. For the extended moving target, the attack strategy is the same as \eqref{HybridAttackStrategyStates} and \eqref{HybridAttackStrategySensors} except that $x_k^A$, $\hat{x}_{k|k-1}$, $y_k^A$, $y_k^a$, $d_k^a$, $w_k$, $v_k$, $A_k$, $B_k$, $C_k$, and $K_k$ are replaced by $\bar{x}_k^A$, $\hat{\bar{x}}_{k|k-1}$, $\bar{y}_k^A$, $\bar{y}_k^a$, $\bar{d}_k^a$, $\bar{w}_k$, $\bar{v}_k$, $\mathcal{A}_k$, $\mathcal{B}_k$, $\mathcal{C}_k$, and $\mathcal{K}_k$, respectively. For the nonlinear moving target, this attack strategy is
\begin{align}
&\text{\small$\begin{bmatrix}
\bar{x}_{k+1}^A \\
\hat{\bar{x}}_{k+1|k}
\end{bmatrix} =
\begin{bmatrix}
\mathcal{A}_k & 0 \\
0 & \mathcal{A}_k(I-\mathcal{K}_k\mathcal{C}_k)
\end{bmatrix}
\begin{bmatrix}
\bar{x}_k^A \\
\hat{\bar{x}}_{k|k-1}
\end{bmatrix} +
\begin{bmatrix}
\bar{w}_k \\
0
\end{bmatrix}$} \\ &\text{\small$\quad\quad\quad +
\begin{bmatrix}
\mathcal{B}_k & \mathcal{B}_k & 0 \\
\mathcal{B}_k & 0 & \mathcal{A}_k\mathcal{K}_k
\end{bmatrix}
\begin{bmatrix}
u_k \\
u_k^a \\
\bar{y}_k^a
\end{bmatrix} -
\begin{bmatrix}
0 \\
\mathcal{A}_k\mathcal{K}_k \begin{bmatrix}
G_kh(\hat{x}_{k|k-1}) \\
0
\end{bmatrix}
\end{bmatrix},$} \nonumber \\
&\text{\small$\bar{y}_k^A =
\begin{bmatrix}
\mathcal{C}_k & 0
\end{bmatrix}
\begin{bmatrix}
\bar{x}_k^A \\
\hat{\bar{x}}_{k|k-1}
\end{bmatrix} +
\begin{bmatrix}
G_kh(x_k^A) \\
0
\end{bmatrix} + \bar{v}_k,$} \\
&\text{\small$\quad\quad\quad\bar{d}_k^a = \mathbb{E}\left[\mathcal{C}_k\hat{\bar{x}}_{k|k-1}+\begin{bmatrix}G_kh(\hat{x}_{k|k-1})\\0\end{bmatrix}\middle|\mathcal{I}_k^A\right] - \bar{y}_k^A.$} \nonumber
\end{align}

Here $\mathcal{I}_k^A$ refers to the information available to the attacker as presented in section II for the hybrid moving target, extended moving target, and nonlinear moving target. These attack strategies are motivated by the following result which states that for the extended moving target, the chosen sensor measurement bias $\bar{d}_k^a$ minimizes the expected value of the $\chi^2$ detection statistic. This illustrates the potential effectiveness of the attack when countered by a $\chi^2$ detector. This result can easily be extended to account for the sensor biases and the measurement residues of the hybrid moving target and the nonlinear moving target.
\begin{theorem}
\label{OptimalBiasTheorem}
Consider a strong adversary who knows $\{\mathcal{C}_j,\mathcal{P}_{j|j-1}\}$ for all $j\in\mathbb{Z}$. Defining $\hat{\bar{x}}_{k|k-1}^e\triangleq\mathbb{E}\left[\hat{\bar{x}}_{k|k-1}\middle|\mathcal{I}_k^A\right]$,
\begin{equation} \small
\bar{d}_k^{a*} \triangleq \argmin_{\bar{d}_k^a}\mathbb{E}\left[g_k(\bar{z}_{k-T+1:k})\middle|\mathcal{I}_k^A\right] = \mathcal{C}_k\hat{\bar{x}}_{k|k-1}^e-\bar{y}_k^A.
\end{equation}
\end{theorem}
\begin{proof}
Observe that
\begin{equation} \small
\medmuskip=1.21mu
\thinmuskip=1.21mu
\thickmuskip=1.21mu
\mathbb{E}\left[g_k(\bar{z}_{k-T+1:k})\middle|\mathcal{I}_k^A\right] = \int_{\vartheta_k} \sum_{i=k-T+1}^k \bar{z}_i^T \Sigma_i^{-1} \bar{z}_i f(\vartheta_k|\mathcal{I}_{k}^{A})\mbox{d}\vartheta_k,
\end{equation}
where $\bar{z}_i=\bar{y}_i^A+\bar{d}_i^a-\mathcal{C}_i\hat{\bar{x}}_{i|i-1}$ and $\vartheta_k\triangleq\hat{\bar{x}}_{[k-T+1|k-T]:[k|k-1]}$. Taking the gradient with respect to $\bar{d}_k^a$ and setting the resulting expression equal to 0, we obtain
\begin{equation} \small
\int_{\vartheta_k} 2\Sigma_k^{-1} (\bar{y}_k^A + \bar{d}_k^a - \mathcal{C}_k\hat{\bar{x}}_{k|k-1}) f(\vartheta_k|\mathcal{I}_{k}^{A})\mbox{d}\vartheta_k = 0.
\end{equation}
Solving for $\bar{d}_k^a$ yields
\begin{equation} \small
\bar{d}_k^a = -\bar{y}_k^A + \mathcal{C}_k\int_{\vartheta_k}\hat{\bar{x}}_{k|k-1} f(\vartheta_k|\mathcal{I}_{k}^{A})\mbox{d}\vartheta_k,
\end{equation}
and the result holds.
\end{proof}

\subsection{Bounds on Attacker's State Estimation}
Given the attack strategies considered in the last section, we now want to characterize a lower bound $Z_k$ on the mean square error matrix of the attacker's estimate of $\hat{\bar{x}}_{k|k-1}$. Since $\hat{\bar{x}}_{k|k-1}^e$ represents the attacker's estimate of $\hat{\bar{x}}_{k|k-1}$, this lower bound $Z_k$ is given by
\begin{equation} \small
\label{ErrorMatrix}
\mathbb{E}\left[(\hat{\bar{x}}_{k|k-1}^e-\hat{\bar{x}}_{k|k-1})(\hat{\bar{x}}_{k|k-1}^e-\hat{\bar{x}}_{k|k-1})^T\middle|\bar{y}_{0:k}^A\right] \geq Z_k.
\end{equation}
To approximate $Z_k$, we leverage conditional posterior Cramer-Rao lower bounds for Bayesian sequences. Unlike the traditional Cramer-Rao lower bound which is limited to unbiased estimators, the Bayesian Cramer-Rao lower bound considers both biased and unbiased estimators. Here we propose using the direct conditional posterior Cramer-Rao lower bound as set forth in \cite{LowerBound} to approximate $Z_k$. The authors here make use of the Bayesian Cramer-Rao lower bound or Van Trees bound derived in \cite{VanTrees} which states that the mean squared error matrix is bounded by the inverse of the Fisher information matrix $\mathcal{I}_k$ as follows
\begin{equation} \small
\label{MSEMatrix}
\mathbb{E}\left[(\underline{x}_k^e-\underline{x}_k)(\underline{x}_k^e-\underline{x}_k)^T\middle|\bar{y}_{0:k}^A\right] \geq \mathcal{I}_k^{-1},
\end{equation}
where $\underline{x}_k^e \triangleq \text{\small$\begin{bmatrix} \bar{x}_k^{e^T} & \hat{\bar{x}}_{k|k-1}^{e^T} \end{bmatrix}^T$}$, $\underline{x}_k \triangleq \text{\small$\begin{bmatrix} \bar{x}_k^{A^T} & \hat{\bar{x}}_{k|k-1}^T \end{bmatrix}^T$}$, and $\bar{x}_k^e$ is the attacker's estimate of $\bar{x}_k^A$. $Z_k$ can be obtained by simply taking the lower right $(n+\tilde{n})\times(n+\tilde{n})$ block of $\mathcal{I}_k^{-1}$. As demonstrated in \cite{LowerBound}, $\mathcal{I}_{k+1}$ can be decomposed into two parts as $\mathcal{I}_{k+1} = \mathcal{I}_{k+1}^D + \mathcal{I}_{k+1}^P$, where $\mathcal{I}_{k+1}^D$ represents the information gained from the new measurements averaged over the a priori distribution and $\mathcal{I}_{k+1}^P$ represents the information contained in the a priori distribution.

To compute $\mathcal{I}_{k+1}$, a particle filter is used to represent the distribution of $\underline{x}_{k+1}$ with the weighted particles $\{\underline{x}_{k+1}^{(j)},\omega_k^{(j)}\}_{j=1}^{\mathcal{L}}$. As shown in \cite{LowerBound}, $\mathcal{I}_{k+1}^D$ can be computed using the following approximation
\begin{equation} \small
\label{IDApproximation}
\mathcal{I}_{k+1}^D \approx \sum_{j=1}^{\mathcal{L}} \omega_k^{(j)} \mathcal{J}_{k+1}^S(\underline{x}_{k+1}^{(j)}),
\end{equation}
where $\mathcal{J}_{k+1}^S(\underline{x}_{k+1}^{(j)})$ is the standard Fisher information matrix with element $(m,n)$ given by
\begin{equation} \small
\label{StandardFIM}
\begin{split}
& \mathcal{J}_{k+1}^S(m,n) = \frac{\partial \mu(\underline{x}_{k+1}^{(j)})^T}{\partial \underline{x}_{k+1}^{(j)}(m)} \Sigma(\underline{x}_{k+1}^{(j)})^{-1} \frac{\partial \mu(\underline{x}_{k+1}^{(j)})}{\partial \underline{x}_{k+1}^{(j)}(n)} \\
& \quad + \frac{1}{2} \text{Tr}\left(\Sigma(\underline{x}_{k+1}^{(j)})^{-1}\frac{\partial \Sigma(\underline{x}_{k+1}^{(j)})}{\partial \underline{x}_{k+1}^{(j)}(m)}\Sigma(\underline{x}_{k+1}^{(j)})^{-1}\frac{\partial \Sigma(\underline{x}_{k+1}^{(j)})}{\partial \underline{x}_{k+1}^{(j)}(n)}\right),
\end{split}
\end{equation}
where $p(\bar{y}_{k+1}^A|\underline{x}_{k+1}^{(j)}) \sim \mathcal{N}(\mu(\underline{x}_{k+1}^{(j)}),\Sigma(\underline{x}_{k+1}^{(j)}))$. Using a simple Gaussian approximation for the prediction distribution $p(\underline{x}_{k+1}|\bar{y}_{0:k}^A) \approx \mathcal{N}(\mu_k,\bar{\Sigma}_k)$ with
\begin{equation} \small
\medmuskip=2.5mu
\thinmuskip=2.5mu
\thickmuskip=2.5mu
\label{PredictionSigmaMu}
\bar{\Sigma}_k = \sum_{j=1}^{\mathcal{L}} \omega_k^{(j)} (\underline{x}_{k+1}^{(j)}-\mu_k)(\underline{x}_{k+1}^{(j)}-\mu_k)^T, ~ \mu_k = \sum_{j=1}^{\mathcal{L}} \omega_k^{(j)}\underline{x}_{k+1}^{(j)},
\end{equation}
$\mathcal{I}_{k+1}^P$ can be approximated by the inverse of the covariance matrix as demonstrated in \cite{LowerBound} so that $\mathcal{I}_{k+1}^P \approx \bar{\Sigma}_k^{-1}$.

By choosing the importance density of the particle filter to be the prior $p(\underline{x}_{k+1}|\underline{x}_k^{(j)})$, the weight update equation derived in \cite{ParticleFilter} simplifies to $\omega_k^{(j)} = \omega_{k-1}^{(j)} p(\bar{y}_k^A|\underline{x}_k^{(j)})$.

A sequential importance sampling algorithm such as that presented in \cite{ParticleFilter} can be used to implement the particle filter, and resampling can be introduced to keep the particle filter from degenerating. Other algorithms presented in \cite{ParticleFilter} such as the auxiliary sampling importance resampling filter and the regularized particle filter can be used to protect the particle filter from sample impoverishment, which is severe in the case of small process noise.

\begin{remark}
Computing a lower bound for the hybrid moving target is described by replacing $\hat{\bar{x}}_{k|k-1}$, $\hat{\bar{x}}_{k|k-1}^e$, $\bar{x}_k^A$, $\bar{x}_k^e$, $\bar{y}_k^A$, $\bar{d}_k^a$, $\mathcal{C}_k$, $\mathcal{P}_{k|k-1}$, and $\bar{z}_k$ in \cref{ErrorMatrix,MSEMatrix,IDApproximation,StandardFIM,PredictionSigmaMu,LowerBoundResult,LowerBoundProof} with $\hat{x}_{k|k-1}$, $\hat{x}_{k|k-1}^e$, $x_k^A$, $x_k^e$, $y_k^A$, $d_k^a$, $C_k$, $P_{k|k-1}$, and $z_k$, respectively, where $\hat{x}_{k|k-1}^e\triangleq\mathbb{E}\left[\hat{x}_{k|k-1}\middle|\mathcal{I}_k^A\right]$ and $x_k^e$ is the attacker's estimate of $x_k^A$.
\end{remark}

\subsection{Bounds on Detection}
The algorithm described above provides a method for computing an approximate lower bound on the mean square error matrix of the attacker's estimate of $\hat{\bar{x}}_{k|k-1}$ for a given set of inputs $u_{0:k}$, $u_{0:k}^a$, $\bar{d}_{0:k}^a$ and observation history $\bar{y}_{0:k}^A$, allowing us to obtain a lower bound on the expected value of the $\chi^2$ detection statistic. The following result characterizes how small an attacker is able to make the detection statistic given the information available to him or her.
\begin{theorem}
Consider a strong adversary who knows $\{\mathcal{C}_j,\mathcal{P}_{j|j-1}\}$ for all $j\in\mathbb{Z}$. Suppose a lower bound $Z_i$ on the error matrix of $\hat{\bar{x}}_{i|i-1}$ is obtained for $i=\{k-T+1,\cdots,k\}$ as presented in (\ref{ErrorMatrix}). Then we have
\begin{equation} \small
\label{LowerBoundResult}
\min_{\bar{d}_k^a}\mathbb{E}\left[g_k(\bar{z}_{k-T+1:k})\middle|\mathcal{I}_k^A\right]\geq\sum_{i=k-T+1}^k\textup{Tr}(\mathcal{C}_{i}^T\Sigma_i^{-1}\mathcal{C}_{i}Z_i).
\end{equation}
\end{theorem}
 \begin{proof}
 We have the following.
 \begin{equation} \footnotesize
\medmuskip=0.14mu
\thinmuskip=0.14mu
\thickmuskip=0.14mu
 \label{LowerBoundProof}
 \begin{split}
 &\min_{\bar{d}_k^a}\mathbb{E}\left[g_k(\bar{z}_{k-T+1:k})\middle|\mathcal{I}_k^A\right] \\
 &=  \mathbb{E}\left[\sum_{i=k-T+1}^k(\bar{y}_i^A+\bar{d}_i^{a*}-\mathcal{C}_i\hat{\bar{x}}_{i|i-1})^T \Sigma_i^{-1}(\bar{y}_i^A+\bar{d}_i^{a*}-\mathcal{C}_i\hat{\bar{x}}_{i|i-1})\middle|\mathcal{I}_k^A\right]= \\
  &  \text{Tr}\left(\mathbb{E}\left[\sum_{i=k-T+1}^k(\mathcal{C}_i(\hat{\bar{x}}_{i|i-1}^e-\hat{\bar{x}}_{i|i-1}))(\mathcal{C}_i(\hat{\bar{x}}_{i|i-1}^e-\hat{\bar{x}}_{i|i-1}))^T \Sigma_i^{-1}\middle|\mathcal{I}_k^A\right]\right) \\
  &= \sum_{i=k-T+1}^k \text{Tr}\left(\mathbb{E}\left[(\hat{\bar{x}}_{i|i-1}^e-\hat{\bar{x}}_{i|i-1})(\hat{\bar{x}}_{i|i-1}^e-\hat{\bar{x}}_{i|i-1})^T\middle|\mathcal{I}_k^A\right]\mathcal{C}_i^T\Sigma_i^{-1}\mathcal{C}_i\right)\\
                       &\ge \sum_{i=k-T+1}^k \text{Tr}(\mathcal{C}_{i}^T \Sigma_i^{-1} \mathcal{C}_i Z_i).
 \end{split}
 \end{equation}
The first three equalities follow from Theorem \ref{OptimalBiasTheorem} and the properties of the trace and expectation. The final inequality follows from \eqref{ErrorMatrix}.
 \end{proof}

\begin{remark}
In general, the adversary's ability to estimate $\{\hat{\bar{x}}_{k|k-1}\}$ is dependent on the inputs $\{u_k^a\}, \{\bar{d}_k^a\}$. For instance, the more the adversary biases the state away from its expected region of operation, the more challenging it is to perform estimation. Thus if the system operator wants to analyze how well an adversary can generate stealthy outputs, he or she must consider a particular sequence of attack inputs $u_k^a, \bar{d}_k^a$.
\end{remark}

\section{Simulation}
We validate each moving target design by considering the quadruple tank process \cite{quadrupletank}, a multivariable laboratory process that consists of four interconnected water tanks. The goal is to control the water level of the first two tanks using two pumps. The system has four states (water level for each tank), two inputs (voltages applied to the pumps), and two outputs (voltages from level measurement devices for the first two tanks). We use an LQG controller with weights following suggestions in \cite{quadrupletankcomparison}. To ensure an appropriate noise magnitude, $Q$, $\tilde{Q}$, $R$, and $\mathcal{R}$ are created by generating a matrix from a uniform distribution, multiplying it by its transpose, and dividing by 100. A window size of 10 is used for the $\chi^2$ detector, $\tilde{A}$ and $\tilde{C}$ are composed of 50\% nonzero entries pulled from a standard normal distribution, and $\tilde{A}$ is stable.

The extended system is comprised of 4 auxiliary states and 2 auxiliary sensors. The auxiliary states can represent the water level of the tub into which each of the tanks dispense water, the rate of change in the tub's water level, and the supply and dispense rates of water flowing into and out of the tub. The auxiliary sensors can measure the tub's water level and the rate at which the tub is supplied with water. The time-varying nature of $\bar{A}_k$, $\tilde{B}_k$, and $\bar{C}_k$ can be achieved by varying the length and width of the tub over time through the auxiliary actuators. The matrices $\tilde{A}$ and $\tilde{C}$ mathematically describe the auxiliary system dynamics. Experiments are averaged over 1000 trials, and simulation results for the hybrid moving target defense and extended moving target defense can be found in \cite{weerakkody2016movingarxiv} and \cite{acc2017griffioen}, respectively.

\subsection{Nonlinear Moving Target Defense}
For the nonlinear moving target, we consider nonlinear functions that take the form of an element-wise power function, $h(x_k)=x_k^c$, $c\in\mathbb{Z}^+$. We first investigate how the power of this nonlinear function affects the amount of information an adversary receives about the time-varying matrices $\bar{A}_k$, $\tilde{B}_k$, and $\bar{C}_k$ through the auxiliary intercepted sensor measurements.

We consider an adversary who starting at time 200 sec. adds a constant input of 0.2 volts to the optimal LQG input and avoids detection by trying to subtract his or her own influence from the sensor measurements as described in (\ref{ZeroDynamics}). We assume that the attacker does not know the realizations of $\bar{A}_k$, $\tilde{B}_k$, $\bar{C}_k$, or $G_k$ but performs his or her attack by sampling the matrices from $\tilde{B}_k(\text{row }i)\sim\mathcal{N}(\mu_{\tilde{B}},\Sigma_{\tilde{B}})$, $\bar{A}_k(\text{row }i)\sim\mathcal{N}(\mu_{\bar{A}},\Sigma_{\bar{A}})$, $\bar{C}_k(\text{row }i)\sim\mathcal{N}(\mu_{\bar{C}},\Sigma_{\bar{C}})$, and $G_k(\text{column i})\sim\mathcal{N}(\mu_G,\Sigma_G)$ where we have chosen $\mu_{\tilde{B}}=\vec{0}$, $\mu_{\bar{A}}=\mu_{\bar{C}}=\vec{1}$, and $\mu_G=\vec{0}$.

We plot the absolute mean tank height deviation in Figure \ref{NonlinearOptimalTankDeviation} for $h(x_k)=x_k^2$ where we see the effect of the attacker's constant bias on the control inputs. In Figure \ref{AFIM}, we plot the spectral norm of the attacker's Fisher information matrix for a few different nonlinear power functions in addition to the case when there is no nonlinear function. For these figures, we use the optimal covariances $\Sigma_{\tilde{B}}^*$, $\Sigma_{\bar{A}}^*$, $\Sigma_{\bar{C}}^*$, and $\Sigma_G^*$ which are generated according to the optimization problems in \eqref{BBkOptimization}, \eqref{ACOptimization}, and \eqref{GkOptimization}. Here the means of $\bar{A}_k$ and $\bar{C}_k$ are used to design a time-invariant $\Sigma_{\tilde{B}}$, and the positive semidefinite bounds are set to $N_B=\vec{1}\vec{1}^T+0.5I$, $N_t=tI$, $\Theta_A=\vec{1}\vec{1}^T+0.5I$, $\Theta_C=\vec{1}\vec{1}^T+0.5I$, $\Theta_i=I$, and $\mathcal{M}=\vec{1}\vec{1}^T+0.5I$.
\begin{figure}[h]
\centering
\subfloat[Tank Height Deviations]{\includegraphics[width=0.5\columnwidth]{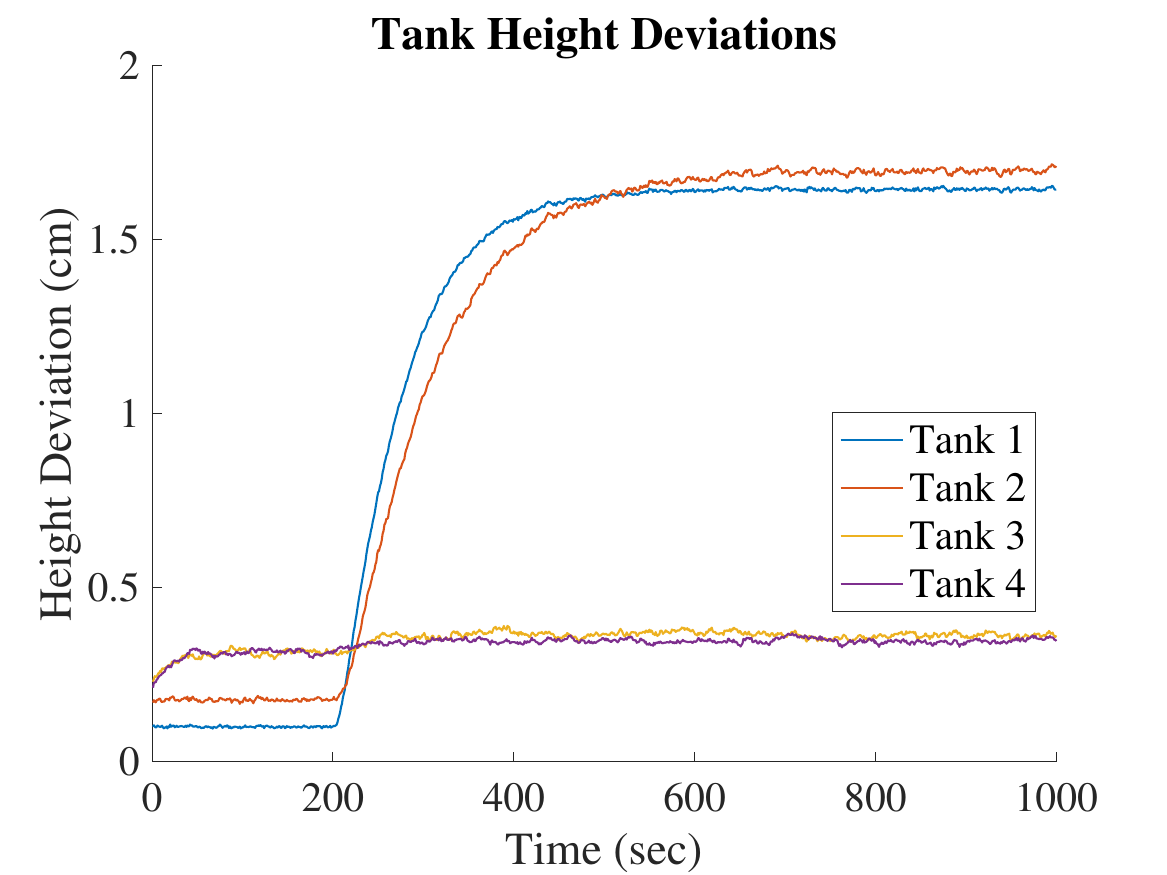}\label{NonlinearOptimalTankDeviation}}
\subfloat[Norm of Fisher Information Matrix]{\includegraphics[width=0.5\columnwidth]{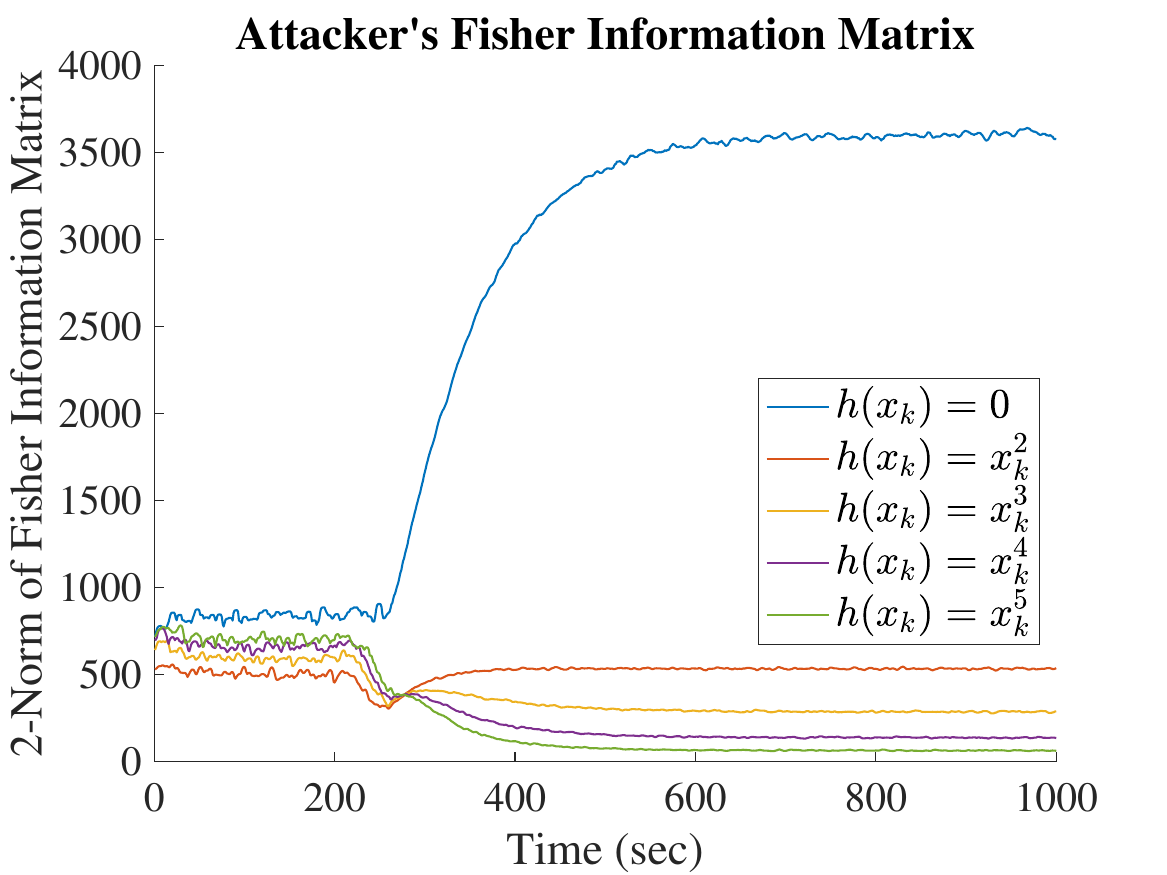}\label{AFIM}}
\caption{For an attacker who adds a constant bias to the control inputs and subtracts his or her influence from the sensor measurements, a) absolute mean tank height deviations and b) spectral norm of the attacker's Fisher information matrix for various nonlinear power functions}
\label{NonlinearFunctions}
\end{figure}

As seen in Figure \ref{AFIM}, the presence of the nonlinearity in the auxiliary sensor measurements results in a decrease of information that the adversary receives about the time-varying matrices $\bar{A}_k$, $\tilde{B}_k$, and $\bar{C}_k$. Furthermore, we see that nonlinear functions with larger powers (which generally have greater magnitudes) cause the adversary to receive less information about the time-varying matrices, consequently making it more difficult for an adversary to generate stealthy outputs.

We now consider the joint design of the covariances $\Sigma_{\bar{A}}$, $\Sigma_{\bar{C}}$, and $\Sigma_G$ for the coupling matrices and nonlinear coefficient matrix of the quadruple tank process. The optimal time-invariant covariance $\Sigma_{\tilde{B}}^*$ obtained previously is used to generate $\tilde{B}_k$, and we consider the same adversary as previously.

In Figure \ref{DetectionStatisticGk}, we plot the detection statistic for optimal and unintelligent designs of $\Sigma_{\bar{A}}$, $\Sigma_{\bar{C}}$, and $\Sigma_G$ with $h(x_k)=x_k^2$. The optimal covariances $\Sigma_{\bar{A}}^*$, $\Sigma_{\bar{C}}^*$, and $\Sigma_G^*$ are generated according to the optimization problems in (\ref{ACOptimization}) and (\ref{GkOptimization}) while the unintelligent covariances $\Sigma_{\bar{A}}^{\text{IID}}=\xi_1^*I$ and $\Sigma_{\bar{C}}^{\text{IID}}=\xi_2^*I$ take IID structures that satisfy all the constraints of (\ref{ACOptimization}) according to
\begin{equation} \small
\label{RandACOptimization}
\begin{split}
& \argmax_{\xi_1,\xi_2} \xi_1 + \xi_2 \quad \text{s.t. } \xi_1I \preceq \Theta_A, \text{ } \xi_2I \preceq \Theta_C, \\
& \quad\quad \text{Tr}(J_{ii})\xi_1I + \text{Tr}(S_{ii})\xi_2I + \text{Sum}(J_{ii})\mu_{\bar{A}}\mu_{\bar{A}}^T \\
& \quad\quad + \text{Sum}(S_{ii})\mu_{\bar{C}}\mu_{\bar{C}}^T + \begin{bmatrix}\mu_{\bar{A}}\text{ }\cdots\text{ }\mu_{\bar{A}}\end{bmatrix}F_{ii}\begin{bmatrix}\mu_{\bar{C}}\text{ }\cdots\text{ }\mu_{\bar{C}}\end{bmatrix}^T \\
& \quad\quad + \begin{bmatrix}\mu_{\bar{C}}\text{ }\cdots\text{ }\mu_{\bar{C}}\end{bmatrix}F_{ii}^T\begin{bmatrix}\mu_{\bar{A}}\text{ }\cdots\text{ }\mu_{\bar{A}}\end{bmatrix}^T \succeq \gamma^* \Theta_i,
\end{split}
\end{equation}
where $i=0,\cdots,T-1$, $\xi_1$ and $\xi_2$ are nonnegative constants, and $\gamma^*$ is the optimal nonnegative constant obtained from (\ref{ACOptimization}). In both (\ref{ACOptimization}) and (\ref{RandACOptimization}), the positive semidefinite bounds are the same as those used previously. The unintelligent covariance $\Sigma_G^{\text{IID}}=\varphi^*I$ takes an IID structure satisfying the first constraint of (\ref{GkOptimization}) according to
\begin{equation} \small
\label{RandGkOptimization}
\argmax_\varphi \varphi \quad \text{s.t. } \varphi I \preceq \mathcal{M},
\end{equation}
where $\varphi$ is a nonnegative constant and the positive semidefinite bound is the same as that used previously. Originally (\ref{RandGkOptimization}) was constructed so that $\Sigma_G^{\text{IID}}$ would satisfy all the constraints in (\ref{GkOptimization}), but this problem proved to be infeasible, implying that the unintelligent covariance is unable to achieve the chosen lower bound on the Fisher information matrix.

As seen in Figure \ref{DetectionStatisticGk}, designing $\Sigma_{\bar{A}}$, $\Sigma_{\bar{C}}$, and $\Sigma_G$ according to (\ref{ACOptimization}) and (\ref{GkOptimization}) results in a detection statistic that is significantly greater than that of a non-optimal design for $\Sigma_{\bar{A}}$, $\Sigma_{\bar{C}}$, and $\Sigma_G$. This supports the idea that increasing the amount of information the defender receives about the attacked states $x_k^A$ through the biased auxiliary sensor measurements $\tilde{y}_k^a$ will result in an increase in detection performance. Even for small biases on the optimal control input (0.2 volts), the nonlinear moving target defense with optimal designs of $\Sigma_{\bar{A}}$, $\Sigma_{\bar{C}}$, and $\Sigma_G$ results in a detection statistic that is far greater than the detection statistic under normal operation.
\begin{figure}[h]
\centering
\subfloat[Detection Statistic]{\includegraphics[width=0.5\columnwidth]{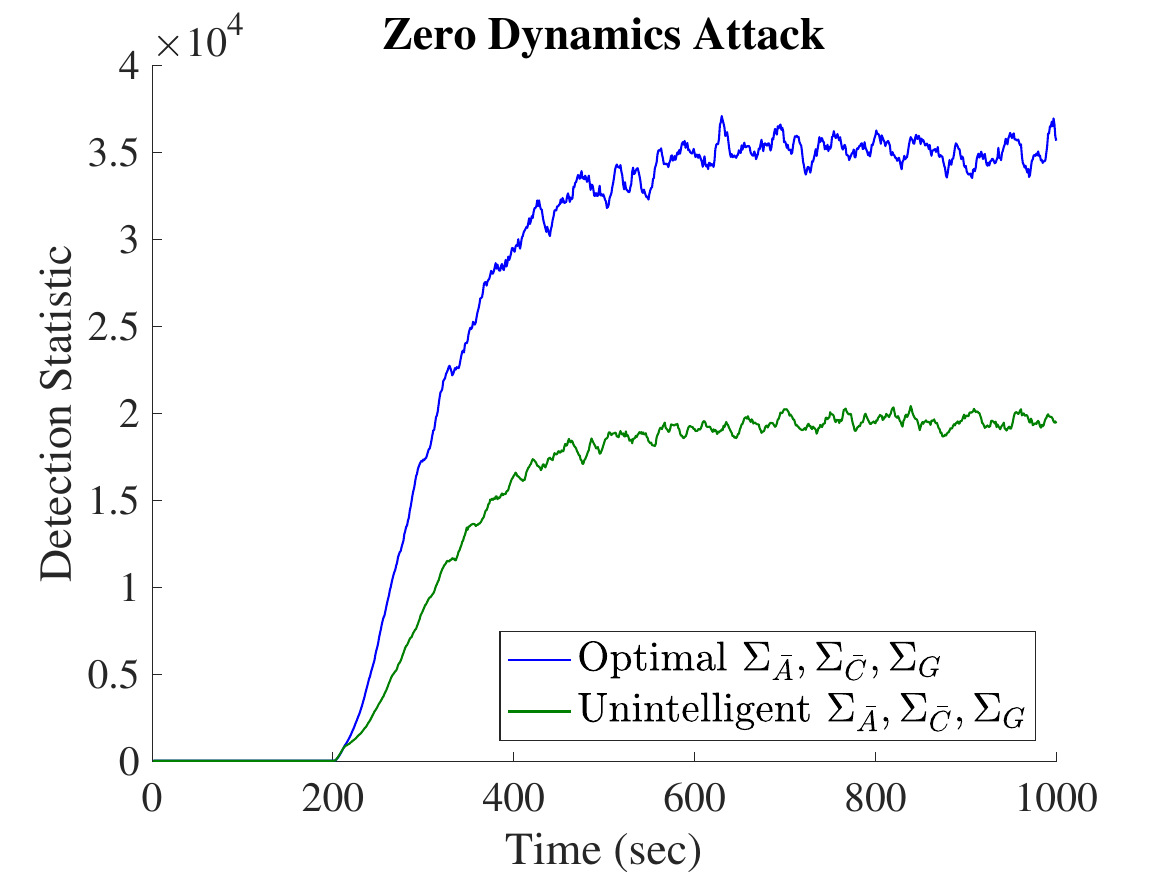}\label{DetectionStatisticGk}}
\subfloat[Norm of Fisher Information Matrix]{\includegraphics[width=0.5\columnwidth]{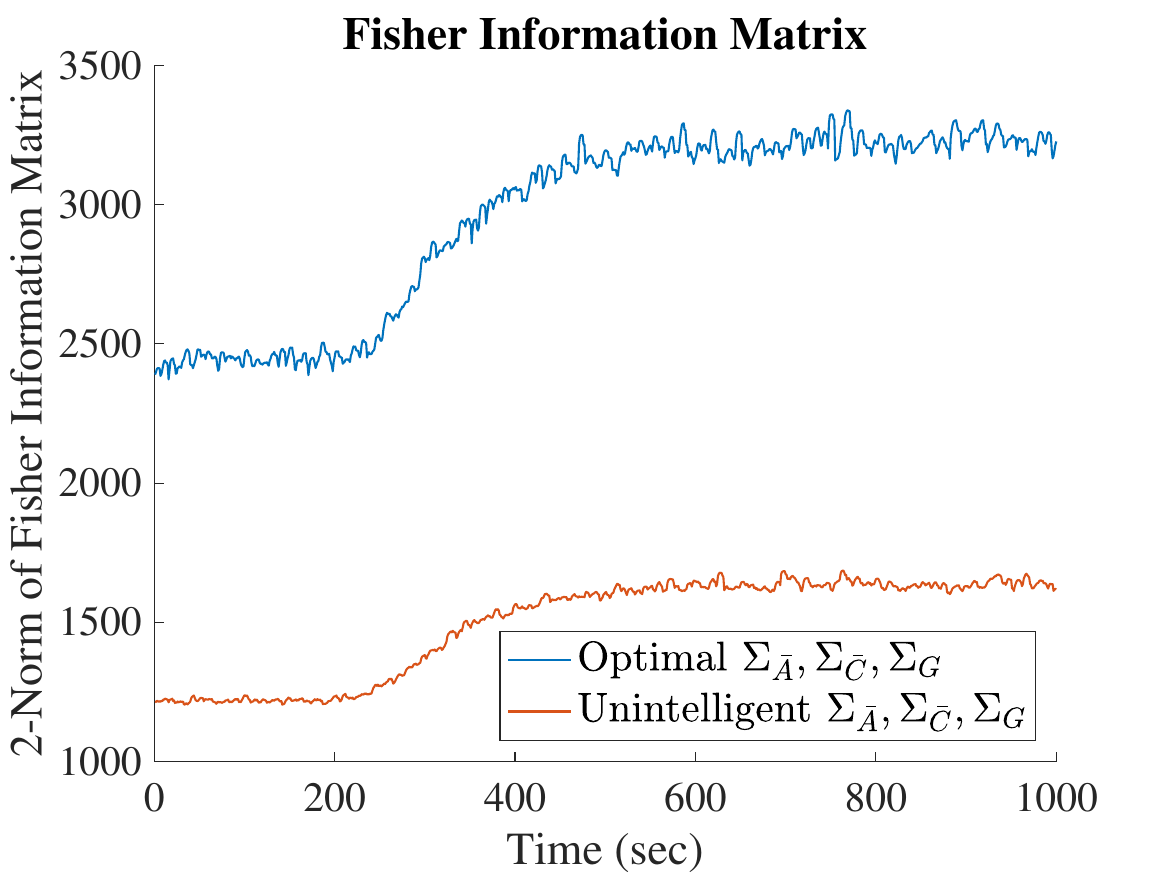}\label{NlnrFIM}}
\caption{For an attacker who adds a constant bias to the control inputs and subtracts his or her influence from the sensor measurements, a) detection statistic for optimal and unintelligent designs of $\Sigma_{\bar{A}}$, $\Sigma_{\bar{C}}$, and $\Sigma_G$ and b) spectral norm of the Fisher information matrix for optimal and unintelligent designs of $\Sigma_{\bar{A}}$, $\Sigma_{\bar{C}}$, and $\Sigma_G$}
\label{NonlinearCoefficientMatrix}
\end{figure}

Figure \ref{NlnrFIM} shows the spectral norm of the Fisher information matrix for optimal and unintelligent designs of $\Sigma_{\bar{A}}$, $\Sigma_{\bar{C}}$, and $\Sigma_G$ with $h(x_k)=x_k^2$. As seen, designing $\Sigma_{\bar{A}}$, $\Sigma_{\bar{C}}$, and $\Sigma_G$ according to (\ref{ACOptimization}) and (\ref{GkOptimization}) results in much more information being gained from the biased auxiliary sensor measurements $\tilde{y}_k^a$ about the attacked states $x_k^A$ than if a non-optimal design were used. Maximizing this amount of information will help produce a more accurate state estimate regardless of whether or not the system is under attack.

\subsection{Bounds on Attacker's Performance}
To investigate lower bounds on the detection statistic when the system is under attack, we consider the extended moving target defense. We consider an adversary who starting at time 200 sec. adds a constant input of 0.3 volts to the optimal LQG input and avoids detection by trying to subtract his or her own influence from the sensor measurements as described in (\ref{ZeroDynamics}). We assume that the attacker does not know the realizations of $\bar{A}_k$, $\tilde{B}_k$, or $\bar{C}_k$ but performs his or her attack by sampling the matrices from $\tilde{B}_k(\text{row }i)\sim\mathcal{N}(\mu_{\tilde{B}},\Sigma_{\tilde{B}})$, $\bar{A}_k(\text{row }i)\sim\mathcal{N}(\mu_{\bar{A}},\Sigma_{\bar{A}})$, and $\bar{C}_k(\text{row }i)\sim\mathcal{N}(\mu_{\bar{C}},\Sigma_{\bar{C}})$ where we have chosen $\mu_{\tilde{B}}=\vec{0}$ and $\mu_{\bar{A}}=\mu_{\bar{C}}=\vec{1}$.

We plot the $\chi^2$ detection statistic and its associated lower bound in Figure \ref{Chi2LowerBound} where we use the optimal covariances $\Sigma_{\tilde{B}}^*$, $\Sigma_{\bar{A}}^*$, and $\Sigma_{\bar{C}}^*$ generated according to the optimization problems in \eqref{BBkOptimization} and \eqref{ACOptimization}. Here the means of $\bar{A}_k$ and $\bar{C}_k$ are used to design a time-invariant $\Sigma_{\tilde{B}}$, and the positive semidefinite bounds are set to $N_B=\vec{1}\vec{1}^T+0.5I$, $N_t=tI$, $\Theta_A=\vec{1}\vec{1}^T+0.5I$, $\Theta_C=\vec{1}\vec{1}^T+0.5I$, and $\Theta_i=I$.
\begin{figure}[h]
\centering
\includegraphics[width=0.5\columnwidth]{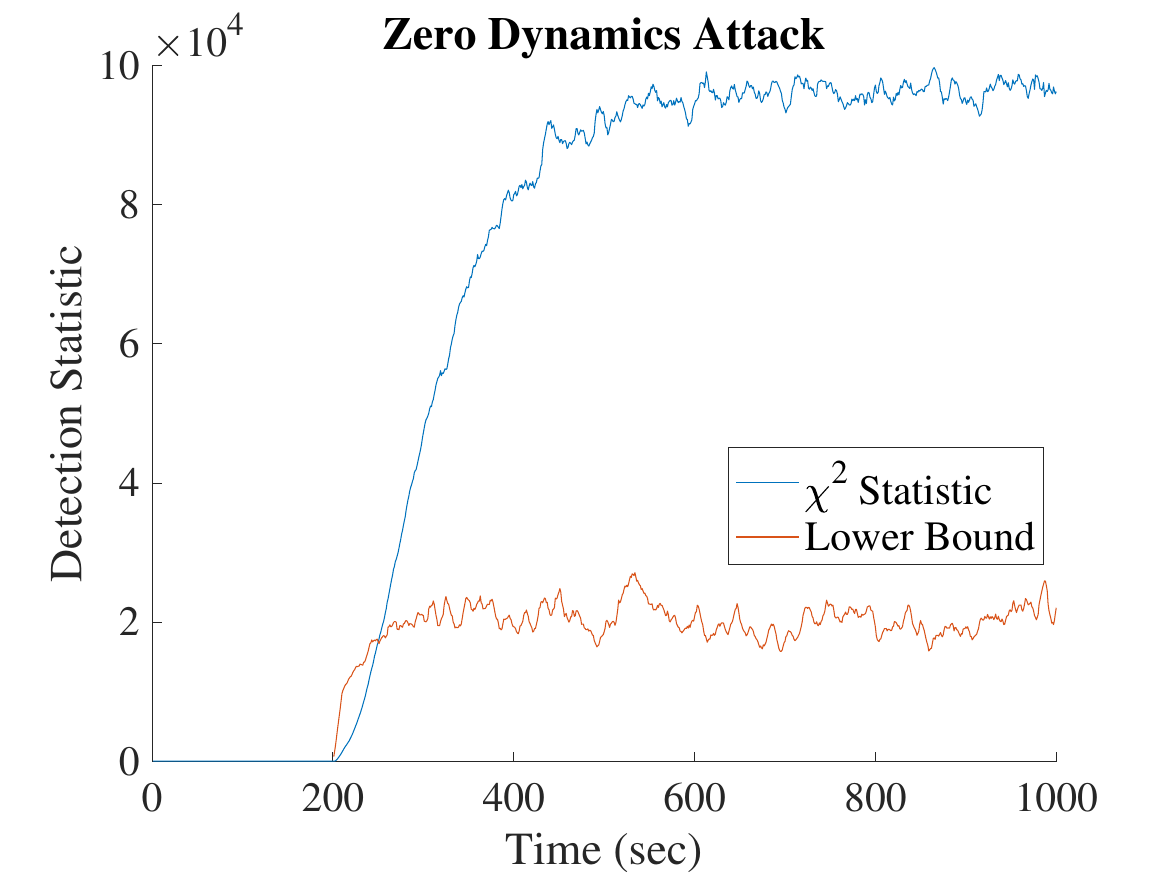}
\caption{$\chi^2$ detection statistic and lower bound on the expected value of the detection statistic for an attacker who adds a constant bias to the control inputs and subtracts his or her influence from the sensor measurements}
\label{Chi2LowerBound}
\end{figure}

As seen in Figure \ref{Chi2LowerBound}, the magnitude of the lower bound on the expected value of the detection statistic is much greater than detection thresholds associated with very small false alarm rates, implying that on average, any attack on this extended moving target system will be detected. These results demonstrate that the moving target defense is effective even in worst case attack scenarios.

\section{Conclusion}
This article presented the moving target defense for detecting and identifying attacks in CPSs. The moving target seeks to limit an adversary's knowledge of the model by introducing stochastic time-varying parameters in the control system. We considered the hybrid moving target, the extended moving target, and the nonlinear moving target, analyzing each system and providing guidelines for the design of the system parameters. We demonstrated how the hybrid moving target enables both detection and identification of malicious nodes, presented designs for the extended moving target that maximize detection and estimation performance, and showed how the nonlinear moving target minimizes any information an adversary receives about the time-varying parameters. Lastly, we investigated lower bounds on the detection statistic, showing that the moving target defense is able to detect even the most stealthy attacks. Future work consists of applying the moving target defense to specific use cases, investigating where and how the time-varying parameters might be introduced to take advantage of the existing system dynamics.

\ifCLASSOPTIONcaptionsoff
  \newpage
\fi

\bibliographystyle{IEEEtran}
\bibliography{root}

\begin{IEEEbiography}[{\includegraphics[width=1in,height=1.25in,clip,keepaspectratio]{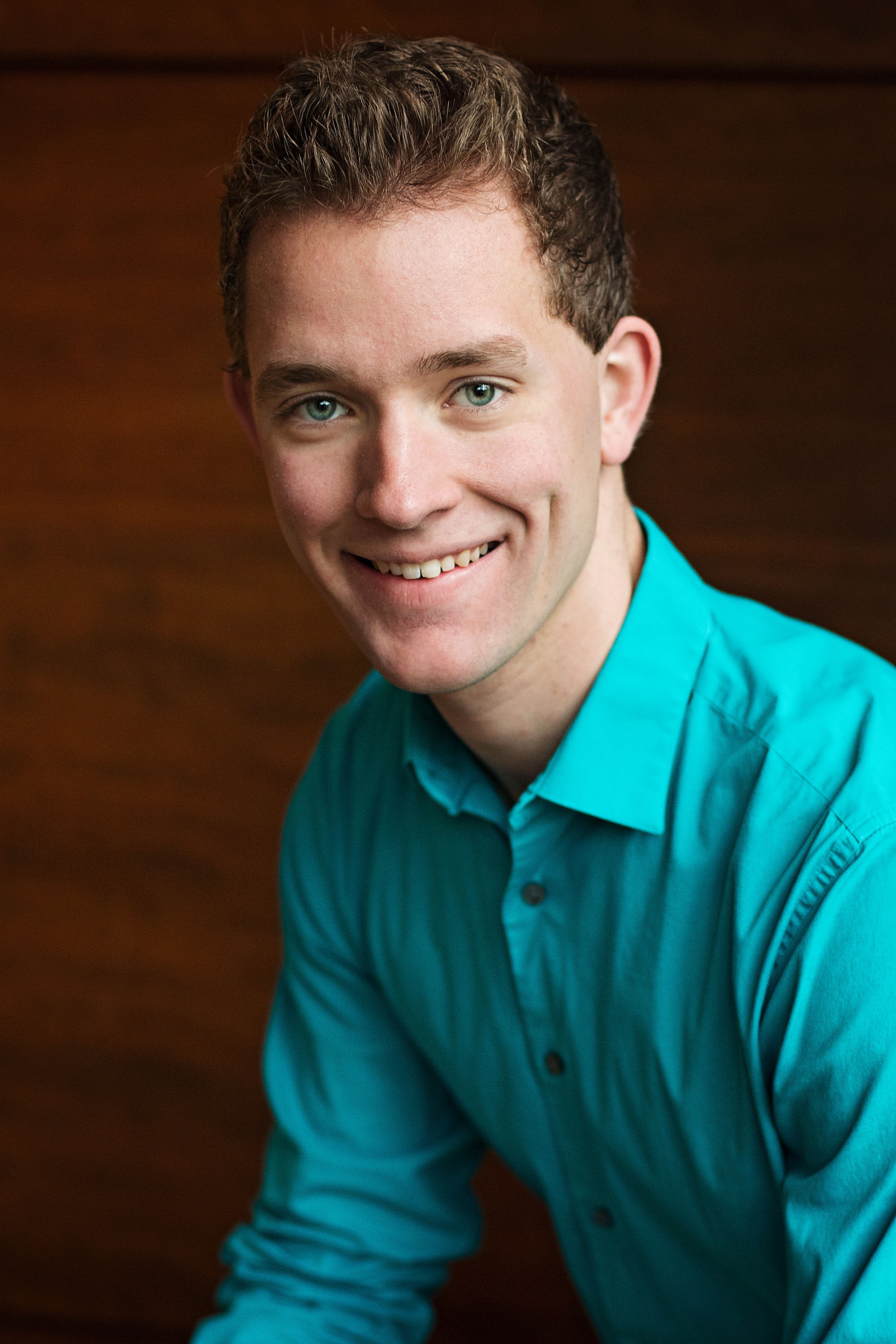}}]{Paul Griffioen}
received the B.S. degree in Engineering, Electrical/Computer concentration, from Calvin College, Grand Rapids, MI, USA in 2016 and the M.S. degree in Electrical and Computer Engineering from Carnegie Mellon University, Pittsburgh, PA, USA in 2018. He is currently pursuing the Ph.D. degree in Electrical and Computer Engineering at Carnegie Mellon University. His research interests include the modeling, analysis, and design of active detection techniques and resilient mechanisms for secure cyber-physical systems.
\end{IEEEbiography}
\begin{IEEEbiography}[{\includegraphics[width=1in,height=1.25in,clip,keepaspectratio]{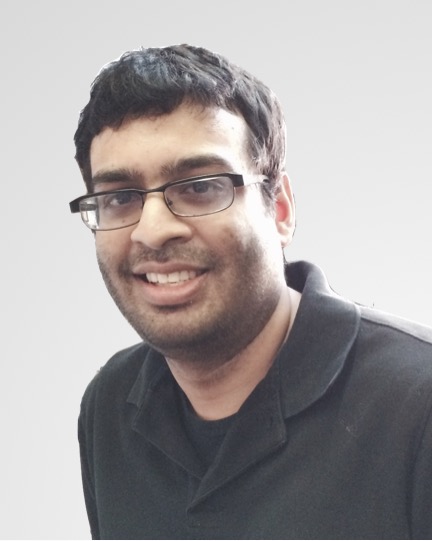}}]{Sean Weerakkody}
received the B.S. degree in Electrical Engineering and Mathematics from the University of Maryland, College Park, USA, in 2012 and the Ph.D. degree in Electrical and Computer Engineering from Carnegie Mellon University, Pittsburgh PA, USA, in 2018. He was awarded the National Defense Science and Engineering Graduate fellowship in 2014 and the Siebel Scholarship in Energy Science in 2018. His research interests include secure design and active detection in cyber-physical systems and estimation in sensor networks.
\end{IEEEbiography}
\begin{IEEEbiography}[{\includegraphics[width=1in,height=1.25in,clip,keepaspectratio]{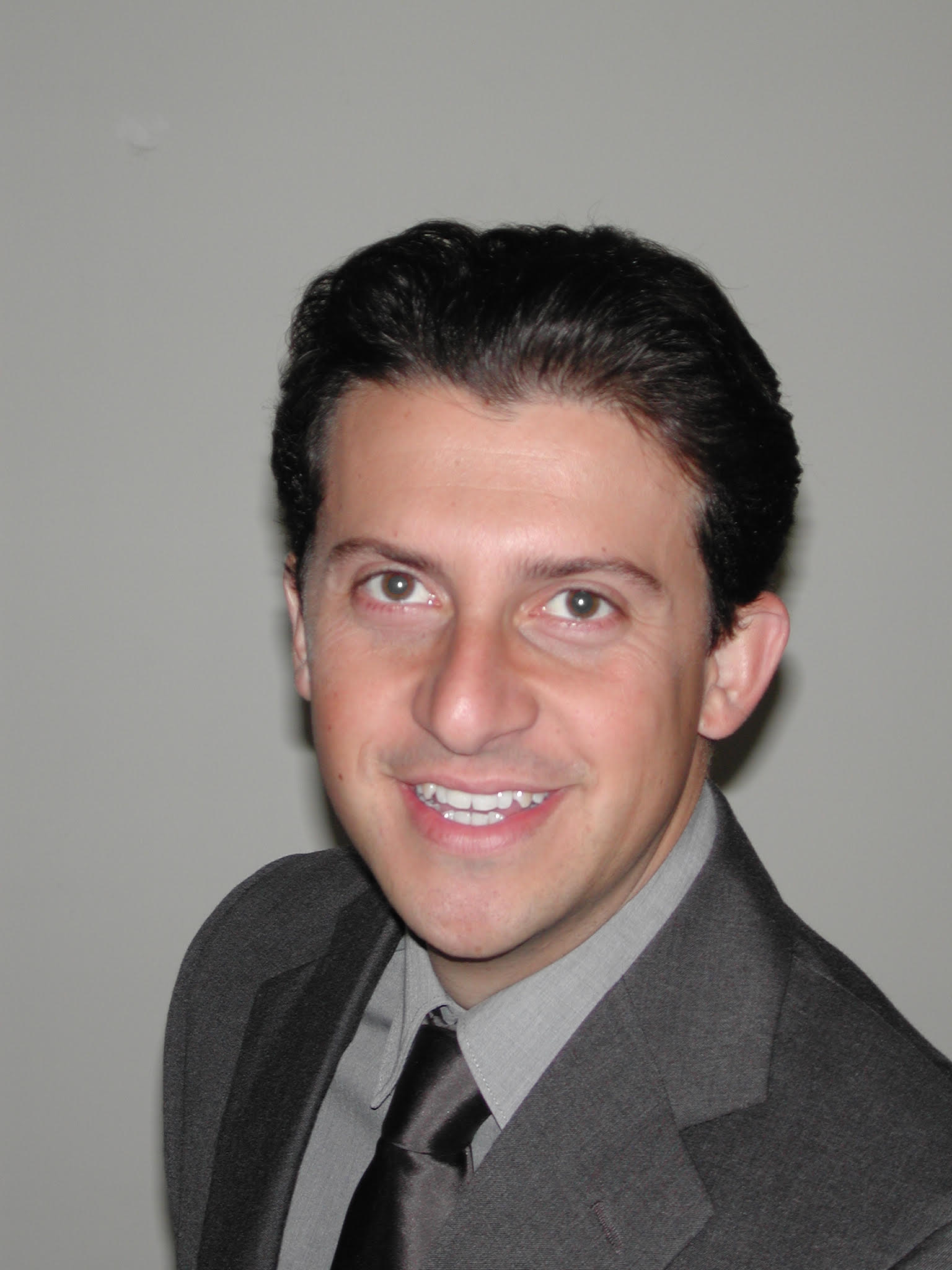}}]{Bruno Sinopoli}
received the Dr. Eng. degree from the University of Padova in 1998 and his M.S. and Ph.D. in Electrical Engineering from the University of California at Berkeley, in 2003 and 2005 respectively. After a postdoctoral position at Stanford University, Dr. Sinopoli was the faculty at Carnegie Mellon University from 2007 to 2019, where he was full professor in the Department of Electrical and Computer Engineering with courtesy appointments in Mechanical Engineering and in the Robotics Institute and co-director of the Smart Infrastructure Institute, a research center aimed at advancing innovation in the modeling analysis and design of smart infrastructure. In 2019 Dr. Sinopoli joined Washington University in Saint Louis, where he is the chair of the Electrical and Systems Engineering department. Dr. Sinopoli was awarded the 2006 Eli Jury Award for outstanding research achievement in the areas of systems, communications, control and signal processing at U.C. Berkeley, the 2010 George Tallman Ladd Research Award from Carnegie Mellon University and the NSF Career award in 2010. His research interests include the modeling, analysis and design of Secure by Design Cyber-Physical Systems with applications to Energy Systems, Interdependent Infrastructures and Internet of Things.
\end{IEEEbiography}

\end{document}